\documentclass[reprint, amsmath,amssymb, aps, onecolumn]{revtex4-2}

\usepackage{graphicx}
\usepackage{dcolumn}
\usepackage{bm}
\usepackage{dsfont}
\usepackage{color}
\usepackage{amsthm}
\definecolor{labelkey}{cmyk}{.4,.2,0,0}
\usepackage{comment}
\usepackage{physics}
\usepackage{bbold}
\usepackage{hyperref}

\newcommand{\1}{\mathbb{1}}

\newcommand{\diff}{{\rm d}}
\newcommand{\Pf}{\textup{Pf}}
\newcommand{\bv}{\bm v}
\newcommand{\RV}{\mathrm{V}}
\newcommand{\erfc}{\textup{erfc}}

\newtheorem{thm}{Theorem}

\newtheorem{cor}[thm]{Corollary}

\newtheorem{prop}[thm]{Proposition}

\numberwithin{equation}{section}

\begin{document}

\preprint{APS/123-QED}

\title{Interpolating non-Hermitian universality classes A and AI$^\dagger$: eigenvalue density and transition regime}

\author{Mark J. Crumpton$^{1,2}$} \email{mark.j.crumpton@bristol.ac.uk}
\author{Francesco Mezzadri$^1$}

\affiliation{$^1$ Department of Mathematics, University of Bristol, Woodland Road, Bristol, BS8 1TW, United Kingdom \\$^2$Heilbronn Institute for Mathematical Research, Bristol, United Kingdom}

\date{\today}

\begin{abstract}
We employ the recently developed Kac-Rice formalism for non-Hermitian random matrices to derive the joint distribution of an eigenvalue and its associated normalised right eigenvector in a Gaussian ensemble that interpolates between complex Ginibre (Class A) and complex symmetric matrices (Class AI$^\dagger$). This distribution is valid at finite matrix size, $N$, for any value of the interpolation parameter $\sigma \in [0,1]$, with $0$ and $1$ corresponding to classes A and AI$^\dagger$ respectively. The marginal distribution for the density of the eigenvalues is derived at finite $N$ and then considered asymptotically as $N \to \infty$. When considering bulk eigenvalues, we recover the standard circular law for all $\sigma$. Furthermore, for edge eigenvalues we find that for fixed $\sigma$, the eigenvalues follow the edge density associated with matrices in Class A. However, a transitional regime is discovered for the interpolation parameter being scaled as $\sigma = 1 - \kappa N^{-1/2}$, where new edge behaviour is observed for the density of eigenvalues - smoothly interpolating two previously known results. This transitional regime and the associated density of eigenvalues is conjectured to be universal for non-Gaussian matrices and we provide numerical evidence in support of this.
\end{abstract}

\maketitle

\section{Introduction}\label{sec:intro}

We introduce a new random matrix ensemble that models the breaking of time reversal symmetry (TRS) in non-Hermitian matrices. This ensemble is constructed such that it interpolates between two non-Hermitian universality classes, represented by complex Ginibre matrices (no TRS) and complex symmetric matrices (positive TRS). Specifically, we consider a Gaussian ensemble of $N \times N$ random matrices constructed as 
\begin{equation}
    X_N \equiv \sqrt{\frac{1 + \tau}{2}} G_N + \sqrt{\frac{1 - \tau}{2}} G_N^T \ ,
    \label{eq:X_def}
\end{equation}
where $G_N$ is a Ginibre matrix, $T$ denotes the matrix transpose and $\tau \in [0,1]$ is an interpolation parameter, such that $\tau = 1 (0)$ yields complex Ginibre (symmetric) matrices. It is useful to view the interpolation parameter in terms $\sigma \equiv \sqrt{1 - \tau^2}$, which we call the symmetry parameter - see Section \ref{sec:main} for more details on this ensemble.

Within this ensemble, we will present a result for the JPDF of an eigenvalue and the associated eigenvector, which is valid at finite $N$ - Theorem \ref{thm:JPDF_z_v}. From this JPDF, we will obtain the associated density of eigenvalues at finite $N$, see Corollary \ref{cor:rho_finite_N}, then assess its asymptotic limit as $N\to \infty$. This analysis shows that the leading order contribution to the density in the bulk is the same across the entire interpolating ensemble, which simply yields the celebrated circular law. On the other hand, in the edge regime we will prove that, for fixed values of the symmetry parameter, the density always follows the behaviour associated with Ginibre matrices (Corollary \ref{cor:rho_edge_SA}). We refer to this as the \textit{strong asymmetry} regime. Whereas, for values of the symmetry parameter asymptotically close to unity, we will present a result that smoothly interpolates the well-known edge density in the GinUE and the recently derived edge density in the SymOE \cite{AFS25}. This proves the presence of a previously unknown regime of \textit{weak asymmetry}, where time reversal symmetry is broken in non-Hermitian matrices. We conjecture that this regime is universal beyond Gaussian matrices and we provide numerical evidence in support of this. These new results constitute the first analytical description of the transition between non-Hermitian universality classes and also serve as the non-Hermitian extension of the works of Mehta \& Pandey \cite{PM83,MP83}, who studied the breaking of TRS in Hermitian matrices in the 1980s. 

\textbf{Background}: The question of universality has been at the heart of Random Matrix Theory (RMT) for many decades now, see \cite{KuijlaarsUniversality} for a review of results. Central to this topic is finding a way of classifying matrices according to mathematical/ physical symmetries and assigning such matrices to universality classes. Celebrated studies in this direction include the early works of Freeman Dyson, who classified closed quantum systems according to their symmetry under time reversal, using real symmetric, complex Hermitian and quaternion self-dual matrices \cite{DysonThreefold} - belonging to the well-known GOE, GUE and GSE respectively. This so called \textit{threefold way} was later generalised by Altland \& Zirnbauer to the \textit{tenfold way} of classifying Hermitian matrices, now applicable to a wider range of physical systems \cite{AZ97}. However, strikingly from an RMT perspective, it turns out that for the ten different symmetry classes, the bulk spectral statistics fall into the same three universality classes associated with the threefold way, with different statistics at the edge \cite{KW17}. The spectral statistics of Hermitian matrices in the different universality classes are well-established, with a range of important results available in  \cite{Wigner57Surm,MehtaBook,TW,Soshnikov99,PS1997,TaoVuUniversality,Johansson12}.

On the other hand, the associated picture of universality classes for non-Hermitian matrices, which describe open quantum mechanical systems \cite{GHS88,KS}, is not as well-analytically developed, but has received a recent boost in attention, partly thanks to the results of \cite{KSUS19,HKKU20}. In \cite{KSUS19}, the authors establish that for non-Hermitian matrices there are 38 possible symmetry classes according to which matrices and quantum systems can be classified. Then, through intensive numerical simulation of 9 of these 38 symmetry classes, a conjecture is proposed in \cite{HKKU20} that only three types of bulk statistics exist for the eigenvalues of non-Hermitian matrices. This conjecture was also recently extended to include edge behaviour - see the works of Akemann et al. \cite{AAKP25}, which also contains a thorough and illuminating discussion of the topics alluded to here. It should be noted at this point that these works build upon earlier theoretical findings, where a classification of non-Hermitian matrices based on symmetry properties was proposed by Bernard \& LeClair \cite{BL02} and then was given associated matrix representations in the work of Magnea \cite{Magnea08}. 

Consequently, it was proposed in \cite{HKKU20} that these 3 non-Hermitian universality classes can be represented using matrices from the complex Ginibre, the complex symmetric and the complex self-dual ensembles, where classes as denoted as A, AI$^\dagger$ and AII$^\dagger$ respectively - with the naming system due to Bernard \& LeClair \cite{BL02}. Matrices within these classes are distinguished according to their symmetry under time-reversal. Specifically, matrices in Class A have no time-reversal symmetry, whereas matrices in Classes AI$^\dagger$ and AII$^\dagger$ are invariant under time-reversal with sign $+1$ and $-1$ respectively. Analytic results for spectral correlations are well-established in the complex Ginibre ensemble, originally studied in \cite{Ginibre}, which have been demonstrated to hold for a wider class of non-Hermitian Wigner matrices, up to a four moment matching condition, in the work of Tao \& Vu \cite{TaoVu2015}. Furthermore, this work has been generalised to prove universality of the correlation kernels at the edge and in the bulk in \cite{CES21, MO24} respectively. Whereas, in comparison, associated results in classes AI$^\dagger$ and AII$^\dagger$ are much less developed. For example, there currently does not exist a closed formula for the density of eigenvalues in class AII$^\dagger$ and the corresponding equation was only very recently found in class AI$^\dagger$ \cite{AFS25}. This discovery comes as a result of a recent impressive development by Fyodorov in \cite{YanKacRice}, where he outlines a new method to obtain the Joint Probability Distribution Function (JPDF) of an eigenvalue and the associated right eigenvector in ensembles of random matrices, via what he calls the ``Kac-Rice method". This method has since been further utilised to study ergodicity of eigenvectors in coupled Ginibre matrices \cite{DF26}.

Within the findings of \cite{AFS25}, they use the finite size density of complex eigenvalues to derive the leading order contributions to the density of bulk and edge eigenvalues. They found that, when comparing classes A and AI$^\dagger$, the bulk eigenvalues follow the standard circular law in both ensembles, but have a noticeable difference in their edge behaviours. This constituted the first clearly observable and quantifiable difference between the two universality classes, for which they also provide convincing numerical evidence that their formula (originally derived for Gaussian matrices) holds universally within class AI$^\dagger$. They also utilise the JPDF of an eigenvalue and eigenvector to derive the associated distribution of the eigenvector self-overlap, which is an important object in measuring observables in quantum scattering \cite{SFPB,FyoSav2012} and characterising eigenvalue stability \cite{CM}. Their results for the self-overlap constitute a second analytic difference between the two universality classes. For a range of analytical results pertaining to the self-overlap one can consult \cite{CM,MC,JNNPZ,BD, FyodorovCMP, ATTZ, FT, WCF24, CFW25, Tarnowski24, CW24, Osman25}.

Other results are available within the classes AI$^\dagger$ and AII$^\dagger$, which we shall now briefly discuss. Firstly, the recent result for the density in class AI$^\dagger$ complements an older result for the density of eigenvalues in complex symmetric matrices at weak non-Hermiticity \cite{SFT99}, which is given as a complicated triple integral. On the other hand, to the best of our knowledge, no analytic results exist in class AII$^\dagger$ for the density, although there has recently been a conjecture proposed for the density of eigenvalues in the edge regime, which has yet to be rigorously proved \cite{KO25}. Away from the density of eigenvalues, results have started to appear comparing other statistical properties across all three universality classes. For example, in \cite{AMP22}, the authors propose a 2D Coulomb gas model with inverse temperature $\beta$ and show that the eigenvalue spacing distributions in classes AI$^\dagger$ and AII$^\dagger$ fit this model with parameter $\beta = 1.4$ and $\beta = 2.6$ respectively - for more information on general non-Hermitian $\beta$-ensembles one can consult \cite{MT2025,AMPT25}. Furthermore, explicit results have recently been achieved for products of complex conjugate pairs of characteristics polynomials \cite{AAKP25}, as well as numerical studies of spacing ratios \cite{SRP20,KW21,BielefeldGroup2026} in all 3 classes.

Given distinct universality classes, a natural question to ask is whether one can propose an interpolating model between the classes, so as to assess transitional regimes and model the breaking of a certain symmetry. Such a question has already been addressed in the Hermitian case through the works of Mehta \& Pandey \cite{PM83,MP83}, who model the breaking of time-reversal invariance in closed systems, through an ensemble that interpolates between the GOE and GUE. This model was also later found to have applications in Brownian motions by Katori \& Tannemura \cite{KT02}. Moreover, the question of interpolating Hermitian and non-Hermitian matrices has been heavily studied through the elliptic Ginibre ensembles, originally introduced in \cite{Girko2,SCSS}. Eigenvalue statistics are well-established in these ensembles as the joint distributions can be expressed as Pfaffian/determinantal point process (depending on if the entries are real, complex or quaternion) \cite{DGIL94,LehmannSommers,FN2}. The elliptic Ginibre ensembles also yield valuable insight into the transition between Hermitian and non-Hermitian matrices, through the so-called regime of weak non-Hermiticity, introduced by Fyodorov, Khoruzhenko \& Sommers \cite{FKS97a,FKS97b,FKS98}, which occurs when the separation of the eigenvalues along the real line is of the same order as the magnitude of the imaginary components. Further studies in the weak non-Hermiticity regime have also yielded universal behaviour of the kernel \cite{ACV,Osman23} and interpolation between Airy and Poisson kernels \cite{Bender,AP}. Additionally, the Kac-Rice method was employed in \cite{YanKacRice} to model the interpolation between the real and complex Ginibre ensemble (both members of class A), which demonstrated the presence of a ``weakly non-real" regime, where the depletion regime present in the real Ginibre ensemble is gradually healed to produce the uniform density observed in the complex Ginibre ensemble. However, to the best of our knowledge, the question of interpolating non-Hermitian universality classes has yet to be considered. 

\textbf{Outline of Paper}: In this work, we aim to address the question of interpolating non-Hermitian universality classes and simultaneously build upon the works of \cite{PM83,MP83,AFS25} to derive the density of eigenvalues in a Gaussian ensemble which interpolates matrices in Classes A and AI$^\dagger$. Our findings are organised as follows:

\begin{itemize}
    \item Section \ref{sec:main}: We outline our model then present our new results at finite matrix size $N$ and asymptotically for large $N$. This is also accompanied by numerical experimentation to verify and illustrate our findings at finite $N$, as well as to support our conjecture of universality of the transitional regime.
    \item Section \ref{sec:finite_proofs}: Application of the Kac-Rice method to prove our main results at finite matrix size, Theorem \ref{thm:JPDF_z_v} and Corollary \ref{cor:rho_finite_N}. 
    \item Section \ref{sec:AA}: Asymptotic analysis of the density of eigenvalues as $N \to \infty$, thus demonstrating the existence of separate regimes of strong and weak asymmetry at the spectral edge.
\end{itemize}

\textbf{Outlook}: Looking forward, there still remains a considerable number of open problems in the world of non-Hermitian matrices with notable consequences. Firstly, the question of quantifying eigenvector non-orthogonality still remains outstanding in the nHIE, as this problem is beyond the scope of our current JPDF since one would require the joint density to also include the left eigenvector. Achieving a distribution of this kind is likely to be within the realms of the Kac-Rice method, but requires further separate development. One should however note that the Kac-Rice method has already been used to quantify eigenvector non-orthogonality in the SymOE, this is due to the fact that left and right eigenvectors are only related via a simple transposition in symmetric ensembles. Secondly, as discussed above, analytic results are scarce in the third universality class of non-Hermitian matrices and so development of theoretical study in this direction should be considered a priority, potentially starting with the verification of the conjectured edge density in \cite{KO25}. One could then consider further separate interpolation models of non-Hermitian ensembles by considering transitions between classes A and AII$^\dagger$ as well as AI$^\dagger$ and AII$^\dagger$.
On the other hand, analytic differences between classes A and AI$^\dagger$ have only so far been observed at the level of the edge density and eigenvector non-orthogonality \cite{AFS25}, this immediately begs the question of quantifying other distinctions across the two classes. Such differences are likely to be observed at the level of correlation functions, distribution of the largest eigenvalue and counting statistics - with results already developed in the Ginibre ensembles \cite{DGIL94,Rider,LACT_MS,ABES}, awaiting comparison with results from other non-Hermitian universality classes.

\section{Statement \& Discussion of Main Results}
\label{sec:main}

We start this Section by defining and reviewing some existing results for the two ensembles we will be interpolating, namely: the \textit{complex Ginibre ensemble} and the \textit{complex symmetric ensemble}, abbreviated as GinUE and SymOE respectively. Note that the nomenclature of these abbreviations is chosen to reflect that the JPDFs of the respective ensembles are invariant under unitary and orthogonal transformations.

On the one hand, we have the GinUE where the entry-wise JPDF is well-known as
\begin{equation}
    \mathcal{P}^{(\text{GinUE})} \left(G_N\right)  = \frac{1}{\pi^{N^2}} \exp \left[ -\text{Tr} \left( G_N G_N^\dagger \right)  \right] \ , 
    \label{eq:jpdf_GinUE}
\end{equation}
such that $G_N$ is an $N \times N$ complex Ginibre matrix, where both the real and imaginary components of each entry are distributed according to a Gaussian, with mean zero and variance 1/2. One can see that this measure is invariant under unitary transformation, i.e. $G_N \to U G_N U^\dagger$ for unitary $U$, which corresponds to broken time reversal invariance. Within this ensemble, the density of eigenvalues in the complex plane reads
\begin{equation}
    \rho^{(\text{GinUE})}_N(z) = \frac{1}{\pi } \, e^{- \vert z \vert^2} \sum_{n=0}^{N-1} \frac{ \vert z \vert^{2n}}{n!} = \frac{1}{\pi} \, \frac{\Gamma\left( N, \vert z \vert^2 \right)}{\Gamma\left( N \right)}  \ , 
    \label{eq:rho_N_GinUE}
\end{equation}
where we have introduced the incomplete $\Gamma$-function, which is defined as
\begin{equation}
    \Gamma\left( N , a \right) \equiv \Gamma\left(N\right) \, e^{-a} \, \sum_{k=0}^{N-1} \frac{a^k}{k!} = \int_a^\infty \diff u\, u^{N-1} e^{-u} = a^{N}\int_1^\infty \diff u\, u^{N-1} e^{-ua} \ .
    \label{eq:incmpl_Gamma}
\end{equation}
Furthermore, asymptotic limits, as $N \to \infty$, of the above density are well-known in the bulk limit, where $z = \sqrt{N}w$, as
\begin{align}
    \rho^{(\textup{GinUE})}_{\text{bulk}}(w) &\equiv \lim_{N \to \infty} \rho^{(\text{GinUE})}_{N}\Big(z = \sqrt{N} w \Big) = \frac{1}{\pi} \Theta\Big[ 1 - |w|^2 \Big] \ ,
    \label{eq:rho_bulk_GinUE}
\end{align}
i.e. the celebrated circular law (with $\Theta[x>0] = 1$ and 0 otherwise) and in the edge limit with
\begin{align}
    \rho^{(\textup{GinUE})}_{\text{edge}}(\eta) &\equiv \lim_{N \to \infty} \rho^{(\text{GinUE})}_{N}\Big(|z| = \sqrt{N} + \eta \Big) = \frac{1}{2 \pi} \erfc\Big(\sqrt{2} \eta\Big) \ ,
    \label{eq:rho_edge_GinUE}
\end{align}
where we have made use of the the complementary error function defined as 
\begin{equation}
    \erfc(x) \equiv 1 - \text{erf}(x) \qquad \qquad \text{with} \qquad \qquad \text{erf}(x) \equiv \frac{2}{\sqrt{\pi}} \int_0^x e^{-t^2} \diff t \ .
    \label{eq:def_erfc}
\end{equation}
In fact, this density has also been considered in the elliptic Ginibre ensembles, with higher order asymptotic expansions and extensions to the real and quaternionic ensembles available in \cite{LR,AB}. For a comprehensive review of results pertaining to the complex Ginibre ensemble and it applications to physics, we invite the reader to consult \cite{BF}.

On the other hand, in the SymOE, one has the following JPDF of entries
\begin{equation}
    \mathcal{P}^{(\text{SymOE})}\left(S_N\right)  = \frac{1}{\pi^{N(N-1)/2}(2 \pi)^N} \exp \left[ - \frac{1}{2} \text{Tr} \left( S_N \overline{S_N} \right)  \right] \ ,  
    \label{eq:jpdf_SymOE}
\end{equation}
such that $S_N$ is an $N \times N$ complex symmetric matrix, $S_N = S_N^T$, with mean zero entries that have variance two on the diagonal and unit variance otherwise. Note that here and throughout, we use the overbar notation to denote the act of complex conjugation. In contrast to the GinUE, it can be seen that this measure is invariant under orthogonal transformation, i.e. $S_N \to O S_N O^T$ for orthogonal $O$, which corresponds to preserved time reversal invariance. 

As mentioned in Section \ref{sec:intro}, it was only very recently that a finite $N$ equation for the density of eigenvalues was discovered in the SymOE \cite[Eq. (9)]{AFS25}. Within this work, the authors employ a scaling such that their limiting density of eigenvalues converges to a circle of radius $\sqrt{2}$. We choose to slightly rescale this result to allow comparisons with the GinUE as defined above, such that both ensembles have a limiting density of the eigenvalues which is uniform in a disk of radius $\sqrt{N}$. In such a scaling, the density of complex eigenvalues is given by
\begin{align}
    \rho^{(\text{SymOE})}_{N}(z) &= \frac{1}{2\pi} \frac{\Gamma(N,|z|^2)}{\Gamma(N)} + \frac{e^{\frac{|z|^2}{2}}}{4\pi}
    \left(\frac{2}{|z|^2}\right)^{\frac{N+1}{2}}
    \gamma\!\left(\frac{N+1}{2}, \frac{|z|^2}{2}\right) \Bigg[\left( N - 1 - |z|^2 \right) \frac{\Gamma(N, |z|^2)}{\Gamma(N)} +  2 \frac{e^{-|z|^2}|z|^{2N}}{\Gamma(N)}
    \Bigg] \ ,
    \label{eq:sym_density_finite_N}
\end{align}
where we have made use of the lower incomplete $\gamma$-function defined as
\begin{equation}
    \gamma(s,x) \equiv \int_0^x \diff t \, t^{s-1} \, e^{-t} \ .
    \label{eq:def_gamma} 
\end{equation}
Furthermore, by construction, it follows that the density of bulk eigenvalues is given by
\begin{align}
    \rho^{(\textup{SymOE})}_{\text{bulk}}(w) &\equiv \lim_{N \to \infty} \rho^{(\text{SymOE})}_{N}\Big(z = \sqrt{N} w \Big) = \frac{1}{\pi} \Theta\Big[ 1 - |w|^2 \Big] \ .
    \label{eq:rho_bulk_SymOE}
\end{align}
The first real analytic difference between these two ensembles that has been observed so far comes when considering the density of eigenvalues at the spectral edge \cite[Eq. (13)]{AFS25}. After applying our chosen rescaling, the density in this regime is given by  
\begin{align}
    \rho^{(\text{SymOE})}_{\text{edge}}(\eta) &\equiv \lim_{N \to \infty} \rho^{(\text{SymOE})}_{N}\Big(z = \sqrt{N} + \eta \Big) \nonumber \\
    &= \frac{1}{4\pi} \erfc\Big(\sqrt{2} \eta\Big) + \frac{1}{4 \sqrt{\pi}} \Big( 1 + \erf(\eta) \Big) \, e^{\eta^2} \left( \sqrt{\frac{2}{\pi}} \, e^{- 2 \eta^2} - \eta \, \erfc(\sqrt{2} \eta)  \right) \ .
    \label{eq:sym_edge_density}
\end{align}
It is the central aim of this paper to derive an interpolating density equation that smoothly connects Eqs. \eqref{eq:rho_edge_GinUE} and \eqref{eq:sym_edge_density}, so as to address the transition between matrices in non-Hermitian universality classes $A$ and $\text{AI}^\dagger$.

With existing results established, we are now in a position to start defining the interpolating ensemble, introduced in Eq. \eqref{eq:X_def}. One can see that when $\tau = 0$ we obtain complex symmetric matrices (class AI$^\dagger$) and when $\tau = 1$ we arrive at complex Ginibre matrices (class A). The reason for this choice of normalisation, regarding the factors of $\tau$, is to impose the following correlation structure on the entries of $X_N$, 
\begin{equation}
    \expval{|X_{ii}|^2}_{X_N} = 1 + \sigma \, , \hspace{1cm}
    \expval{|X_{ij}|^2}_{X_N} = 1 \, , \hspace{1cm} 
    \expval{X_{ij}\, X_{ji}}_{X_N} = 0 \ , \hspace{1cm} 
    \expval{X_{ij}\, \overline{X_{ji}}}_{X_N} = \sigma \ ,
    \label{eq:correlation_structure}
\end{equation}
for $i \neq j$, where $\expval{\cdot}_{X_N}$ represents an ensemble average with respect to $X_N$ and $\sigma \equiv \sqrt{1 - \tau^2}$. Throughout this work, all results will be expressed in terms of $\sigma$ and we will refer to this as the \textit{symmetry parameter}. Given that the variance of the off-diagonal entries is constant at unity, this ensures that, for all $\tau$, the eigenvalues always converge to the circular law with radius $\sqrt{N}$ in the large $N$ limit (this statement will be proved fully in Section \ref{cor:rho_bulk}). In what follows we shall refer to matrices $X_N$ as belonging to the \textit{non-Hermitian Interpolating Ensemble} (nHIE), with the JPDF of such matrices being given by
\begin{equation}
    \mathcal{P}^{\text{(nHIE)}}(X) = \mathcal{C}_N \, \text{exp}\left[ - \frac{1}{1 - \sigma^2} \text{Tr} \left(X^\dagger X - \sigma \overline{X} X\right) \right] \ ,
    \label{eq:nHIE_JPDF}
\end{equation}
with normalisation constant
\begin{equation}
     \mathcal{C}_N = \frac{1}{\pi^{N^2} (1 +\sigma)^N  \sqrt{1 - \sigma^2}^{N(N-1)}} \ ,
\end{equation}
for $\sigma < 1$. In a similar way to the elliptic Ginibre ensembles, one can carefully take the limit $\sigma \to 1^-$ and obtain Eq. \eqref{eq:jpdf_SymOE}, we leave this as an exercise to the reader.

We now move onto introducing our central findings for the nHIE. Within this ensemble we are able to derive the following result for the JPDF of an eigenvalue and the associated normalised eigenvector.

\begin{thm} \label{thm:JPDF_z_v}
    Let $X_N$ be a random $N \times N$ matrix in the nHIE, defined in Eq. \eqref{eq:X_def}, with symmetry parameter $\sigma \in [0,1]$. Let $z$ be an eigenvalue of $X_N$ and let $\bm v$ be the associated normalised right-eigenvector, such that $X_N \bm v = z \bm v$ and $\bm v^\dagger \bm v = 1$. Then, it follows that the JPDF of $z$ and $\bm v$ is given by 
    \begin{align}
        \mathcal{P}^{(\textup{nHIE})}_{N}(z, \bm v; \sigma) &= \frac{e^{ \frac{\sigma}{1 + \sigma}|\nu|^2|z|^2} \Gamma(N-1)}{\pi^{N+1} (1 + \sigma)}  \Bigg[ (N - 1) \mathcal{S}_{N-1}(|z|^2; \sigma) + |z|^2 (\sigma - 2) \sigma \, \mathcal{S}_{N-2}(|z|^2; \sigma) P^2 \nonumber\\
        &- \Big( (N - 1 ) \mathcal{S}_{N-1}(|z|^2; \sigma) - (1 - \sigma^2) \mathcal{T}_{N - 1}(|z|^2; \sigma) \Big) P^2 + \sigma^2 |z|^2 \mathcal{S}_{N-2}(|z|^2; \sigma) P^4 \Bigg]\ ,
        \label{eq:JPDF_z_v_nHIE}
    \end{align}
    where $\nu \equiv \bm v^T \bm v$ and $P \equiv 1- |\nu|^2$ and we have utilised the functions
    \begin{equation}
        \mathcal{S}_{N}(|z|^2;\sigma) \equiv \sum_{k=0}^N \sigma^{k} \frac{\Gamma(N-k+1, |z|^2)}{\Gamma(N-k+1)} \hspace{1cm} \text{and} \hspace{1cm} \mathcal{T}_{N}(|z|^2;\sigma) \equiv \sum_{k=0}^N \sigma^{k} \frac{(N-k) \Gamma(N-k+1, |z|^2)}{\Gamma(N-k+1)} \ , 
        \label{eq:def_SN_TN}
    \end{equation}
    given in terms of the incomplete $\Gamma$-function defined in Eq. \eqref{eq:incmpl_Gamma}. 
\end{thm}

With the JPDF of the eigenvalue and eigenvector in place, we can now apply an integration result from \cite{AFS25} to derive the marginal density of the eigenvalues as follows:

\begin{cor}\label{cor:rho_finite_N}
    For random $N\times N$ matrices in the nHIE, Eq. \eqref{eq:X_def}, the density of eigenvalues in the complex plane follows directly from Theorem \ref{thm:JPDF_z_v} as:
    \begin{align}
        \rho^{(\textup{nHIE})}_{N}(z;\sigma) &= \frac{1}{\pi (1+\sigma)} \mathcal{S}_{N-1}(|z|^2;\sigma) - \frac{\sigma}{2\pi}  \mathcal{S}_{N-2}(|z|^2;\sigma) + \frac{e^{\frac{\sigma}{1+\sigma}|z|^2}}{2\pi(1+\sigma)}
        \left(\frac{1+\sigma}{\sigma |z|^2}\right)^{\frac{N+1}{2}}
        \gamma\!\left(\frac{N+1}{2}, \frac{\sigma}{1+\sigma}|z|^2\right) \nonumber \\
        & \times \Bigg[ \frac{(1+\sigma)(N+1)}{2} \sigma \, \mathcal{S}_{N-2}(|z|^2;\sigma) - \Big(  ( N-1 ) \mathcal{S}_{N-1}(|z|^2;\sigma) -  (1-\sigma^2)\, \mathcal{T}_{N-1}(|z|^2;\sigma) \Big) \nonumber\\
        &+ |z|^2 \frac{\sigma}{1+\sigma} \left( \sigma \, \mathcal{S}_{N-1}(|z|^2;\sigma)  + (\sigma-2)\left( \mathcal{S}_{N-2}(|z|^2;\sigma) -  \frac{\Gamma(N, |z|^2)}{\Gamma(N)} \right) \right)
        \Bigg] \ ,
        \label{eq:rho_finite_N}
    \end{align}
    where we have made use of the lower incomplete $\gamma$-function defined in Eq. \eqref{eq:def_gamma}. Note that this density is normalised such that $\int_{\mathbb{C}} \diff^2 z \, \rho^{(\textup{nHIE})}_{N}(z;\sigma) = N$.
\end{cor}

Considering the form of Eq. \eqref{eq:rho_finite_N}, one can notice that the density of eigenvalues is rotationally invariant, i.e. it only depends on $|z|$, this in accordance with previous observations for the GinUE and SymOE, Eqs. \eqref{eq:rho_N_GinUE} and \eqref{eq:sym_density_finite_N} respectively. This is not unexpected as all members of the interpolating ensemble are complex matrices, thus effects that can break rotational symmetry, such as eigenvalue depletion close to the real line (which is typically seen in real and quaternionic matrices), are not present. Furthermore, one can also see that Eq. \eqref{eq:rho_finite_N} readily reproduces the Ginibre density, Eq \eqref{eq:rho_N_GinUE}, since $\gamma(M,0) =0$ for all $M>0$ by Eq. \eqref{eq:def_gamma} and 
\begin{equation}
    \mathcal{S}_{N-1}(|z|^2;0) = \frac{\Gamma(N,|z|^2)}{\Gamma(N)} \ ,
\end{equation}
according to Eq. \eqref{eq:def_SN_TN}. On the other hand, Eq. \eqref{eq:rho_finite_N} reproduces Eq. \eqref{eq:sym_density_finite_N} when $\sigma = 1$, this can be seen through the fact that 
\begin{equation}
    \mathcal{S}_N(|z|^2; \sigma = 1) = \sum_{k=0}^N \frac{\Gamma(k+1,|z|^2)}{\Gamma(k+1)} = \frac{\Gamma(N+2,|z|^2)(N+1-|z|^2) + e^{-|z|^2} |z|^{2(N+2)}}{\Gamma(N+2)} \ ,
    \label{eq:S_sym_no_sum}
\end{equation}
which can be derived by employing the integral representation of the incomplete $\Gamma$-function after the first equality.

In Figure \ref{fig:finite_N}, we show the distribution of $|z|$ for fixed $N$ and fixed $\sigma$ separately. The fixed $N$ plot illustrates how the nHIE links the GinUE and SymOE with numerical results shown for $\sigma = 0.25$ and $\sigma =0.75$. Whereas the fixed $\sigma = 0.5$ plot shows how the standard triangular law (i.e. $\rho(|z|) = |z| \Theta[\sqrt{N} - |z|]$ which follows from the circular law) for the density of eigenvalues forms in the nHIE. In both of these plots, the distribution is normalised to one, which is obtained by multiplying the corresponding density by $2 \pi |z|/ N$. 

\begin{figure}[h]
    \centering
    \includegraphics[width=0.48\linewidth]{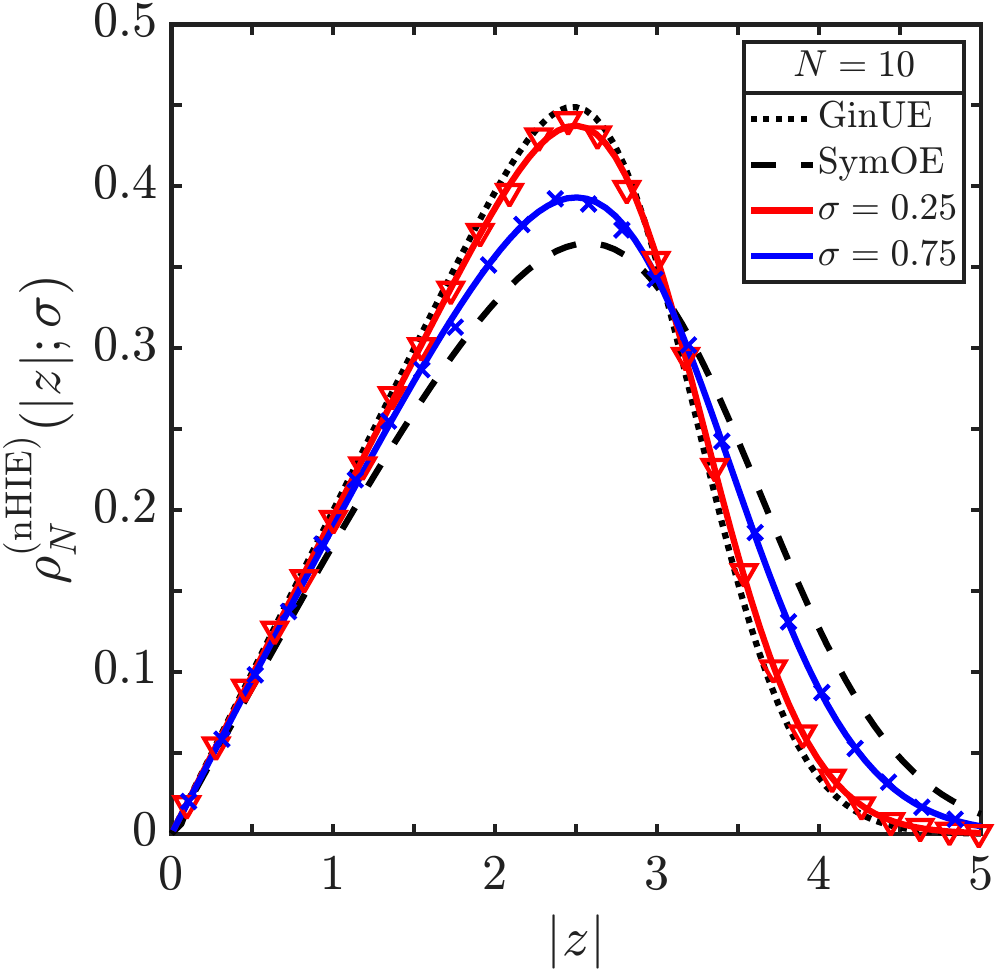}
    \hspace{0.5cm}
    \includegraphics[width=0.48\linewidth]{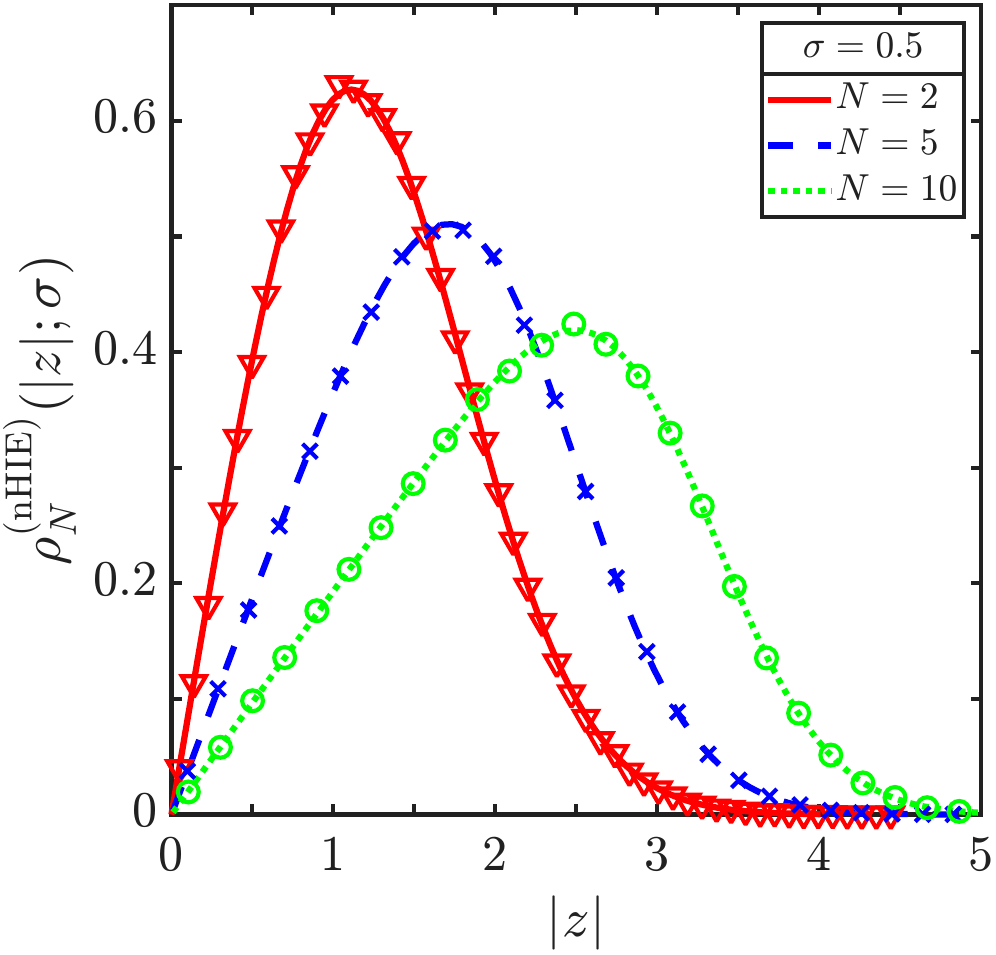}
    \caption{Normalised distribution of eigenvalue moduli in the nHIE at finite $N$ and a comparison to the GinUE and SymOE. Left: Effect of changing $\sigma$ for a fixed $N$. Right: Effect of changing $N$ for a fixed $\sigma$. Theoretical lines for the GinUE, SymOE and nHIE are plotted using appropriately normalised versions of Eqs. \eqref{eq:rho_N_GinUE}, \eqref{eq:sym_density_finite_N} and \eqref{eq:rho_finite_N} respectively.  In each plot coloured markers are generated from the spectra of 100,000 randomly generated matrices. }
    \label{fig:finite_N}
\end{figure}

Now that we have an equation in place for the density of eigenvalues which is valid at finite $N$, the natural and interesting question that follows is to consider its asymptotic limit as $N \to \infty$. As with matrices in the GinUE, one must consider two separate scaling regimes for the positions of the eigenvalues. These are: the \textit{bulk} (where $z = \sqrt{N} w$) and the \textit{edge} (with $|z| = \sqrt{N} + \eta$) regimes, defined for finite $w$ and $\eta$, with the extra constraint that $\eta \in \mathbb{R}$. Analysing the bulk regime first, we find the following:

\begin{cor}\label{cor:rho_bulk}
    Consider an eigenvalue in the bulk of the nHIE, scaled as $z = \sqrt{N} w$. Then, it follows from Corollary \ref{cor:rho_finite_N}, that for all $\sigma \in [0,1]$, we have that
    \begin{align}
        \rho^{(\textup{nHIE})}_{\textup{bulk}}(w) \equiv \lim_{N \to \infty} \rho^{(\textup{nHIE})}_{N}\big( \sqrt{N} w;\sigma \big) 
        = \frac{1}{\pi} \Theta[1 - |w|^2] \ ,
        \label{eq:rho_nHIE_bulk_cor}
    \end{align}
    where $\Theta[x]$ is the standard Heaviside $\Theta$-function.
\end{cor}

It is well-known that the circular law holds for the GinUE in such a scaling \cite{BC2012}. It was also recently shown in \cite{AFS25} that the circular law also applies to leading order in the bulk of an ensemble of Gaussian complex symmetric matrices. Therefore, it is potentially unsurprising that the circular law applies for all $\sigma \in [0,1]$ in the nHIE, yet it must be proved mathematically, which we do in Section \ref{subsec:proof_bulk}.

On the other hand, in the edge scaling regime, distinct analytic differences have already been observed between the two universality classes for the density of eigenvalues, cf. Eqs. \eqref{eq:rho_edge_GinUE} and \eqref{eq:sym_edge_density}. A central goal of this paper was to interpolate these two results to address the transition between matrices in classes A and AI$^\dagger$ (i.e. the breaking of time reversal invariance). In order to do this, we must introduce two separate scaling regimes of the symmetry parameter $\sigma$. The first, and simplest, of these scaling regimes arises for a fixed $\sigma \in [0,1)$ as $N \to \infty$, this is referred to as the \textit{strong asymmetry} (SA) regime. In this regime, we find the following for the density of eigenvalues:

\begin{cor}\label{cor:rho_edge_SA}
    Let $z$ be an eigenvalue at the edge of the spectrum of the nHIE at strong asymmetry, such that $|z| = \sqrt{N} + \eta$ and $\sigma \in [0,1)$ is fixed. Employing such a scaling in Corollary \ref{cor:rho_finite_N} and taking the limit of $N \to \infty$, one finds that
    \begin{align}
        \rho^{(\textup{nHIE})}_{\textup{edge,SA}}(\eta) &\equiv \lim_{N \to \infty} \rho^{(\textup{nHIE})}_{N}\Big(|z| = \sqrt{N} + \eta ; \sigma \Big) = \frac{1}{2 \pi} \erfc\Big(\sqrt{2} \eta\Big) \ .
        \label{eq:rho_edge_SA}
    \end{align}
\end{cor}

The result of Corollary \ref{cor:rho_edge_SA} demonstrates that, for complex matrices that are not asymptotically close to being symmetric, the edge density follows the same leading order behaviour as Eq. \eqref{eq:rho_edge_GinUE}, i.e. they fall into the Ginibre universality class. This Corollary is proved in Section \ref{subsec:proof_edge_SA} and has been tested numerically in FIG. \ref{fig:edge_densities} for $\sigma = 0.5$ and $N=2500$, with very strong agreement.

The findings of the previous corollary motivates the need for the second edge scaling regime, which arises when the symmetry parameter becomes asymptotically close to unity, such that $\sigma = 1 - \kappa N^{-1/2}$ (with non-negative $\kappa$). In such a regime, matrices are symmetric to leading order with a sub-leading anti-symmetric correction, such that
\begin{equation}
    X_N = \frac{G_N + G_N^T}{\sqrt{2}} + \frac{1}{N^{1/4}} \sqrt{\frac{\kappa}{2}} \frac{G_N - G_N^T}{\sqrt{2}} + O\left(N^{-3/2}\right) \ ,
\end{equation}
which follows from the fact that, in this scaling, $\tau = \sqrt{ 2 \kappa/\sqrt{N} - \kappa^2/N }$. Note that here and throughout, the $O(m^a)$ notation indicates terms are at most proportional to $m^a$ as $m$ becomes large. In what follows, we will refer to this scaling of $\sigma$ as the \textit{weak asymmetry} (WA) regime, which leads directly to the following Corollary:

\begin{cor}\label{cor:rho_edge_WA}
    Consider an eigenvalue $z$ in the edge regime of the spectrum of nHIE matrices at weak asymmetry, where $|z| = \sqrt{N} + \eta$ and $\sigma = 1 - \kappa N^{-1/2}$ with $\kappa > 0$. The use of such a scaling in Corollary \ref{cor:rho_finite_N}, for $N \to \infty$, leads to
    \begin{align}
    \rho^{(\textup{nHIE})}_{\textup{edge,WA}}(\eta;\kappa) &\equiv \lim_{N \to \infty} \rho^{(\textup{nHIE})}_{N}\left(|z| = \sqrt{N} + \eta ; \sigma = 1 - \frac{\kappa}{\sqrt{N}} \right) \nonumber\\
    &= \frac{1}{\pi}\left[ \frac{3}{8} \erfc(\sqrt{2} \eta) - \frac{1}{8} e^{2 \eta \kappa + \frac{\kappa^2}{2}} \erfc \left( \frac{2 \eta + \kappa}{\sqrt{2}} \right) \right] + \frac{1}{4 \sqrt{\pi}} \, e^{\left( \eta - \frac{\kappa}{4} \right)^2} \, \bigg( 1+ \erf\left( \eta - \frac{\kappa}{4} \right) \bigg) \nonumber\\
    & \times\left[ \sqrt{\frac{2}{\pi}} \, e^{-2 \eta^2} - \frac{1}{2} \left( \eta - \frac{\kappa}{4} \right) \erfc\left( \sqrt{2} \eta \right)- \frac{4 \eta + 3\kappa}{8} \, e^{2\eta\kappa + \frac{\kappa^2}{2}} \erfc\left( \frac{2 \eta + \kappa}{\sqrt{2}} \right) \right] \ .
    \label{eq:rho_WA_edge}
\end{align}
\end{cor}

One can easily see that the above result reproduces the edge density in the class AI$^\dagger$, Eq. \eqref{eq:sym_edge_density}, when we set $\kappa = 0$. Furthermore, it can also be seen that when $\kappa\to \infty$, one obtains Eq. \eqref{eq:rho_edge_SA} - thus Eq. \eqref{eq:rho_WA_edge} correctly produces the two endpoints of the nHIE, without the need for an intermediary scale different from $\sigma = 1 - \kappa N^{-1/2}$. This result of Corollary \ref{cor:rho_edge_WA} is proved in detail in Section \ref{subsec:proof_edge_WA} and, as part of this proof, we also show the reduction to Eq. \eqref{eq:rho_edge_SA} when $\kappa \to \infty$. Additionally, we test this equation numerically for Gaussian random matrices in FIG \ref{fig:edge_densities}, observing a strong agreement. 

Furthermore, as seen in a range of other notable works, it often appears to be the case that a leading order result that can be rigorously derived for Gaussian matrices, holds universally for much larger classes of random matrices. See the following works for examples where Gaussian results are shown to be universal through numerical experiments and mathematical proofs \cite{PS1997,BD,CES21,Osman25}. To this end, we conjecture that Eqs. \eqref{eq:rho_edge_SA} and \eqref{eq:rho_WA_edge} hold universally for strongly and weakly asymmetric non-Gaussian matrices $X_N$ generated according to Eq. \eqref{eq:X_def}, providing that such matrices satisfy the correlation structure in Eq. \eqref{eq:correlation_structure}.

To support this conjecture, we numerically generate matrices using two different non-Gaussian entry distributions for the matrices $G_N$ in Eq. \eqref{eq:X_def}. Firstly, we consider complex entries with uniformly distributed real and imaginary parts on $[- \sqrt{3/2}, \sqrt{3/2}]$. Secondly, we use two independent \textit{Bernoulli} random variables $x_{ij}$ and $y_{ij}$ (which take values $\pm 1$ with equal probability) and construct the entries of $G_N$ as $G_{ij} = (x_{ij} + i y_{ij})/\sqrt{2}$. The results, which seem to strongly support our conjecture, of this numerical experimentation are shown in Figure \ref{fig:edge_densities}. \\

\begin{figure}
    \centering
    \includegraphics[width=0.48\linewidth]{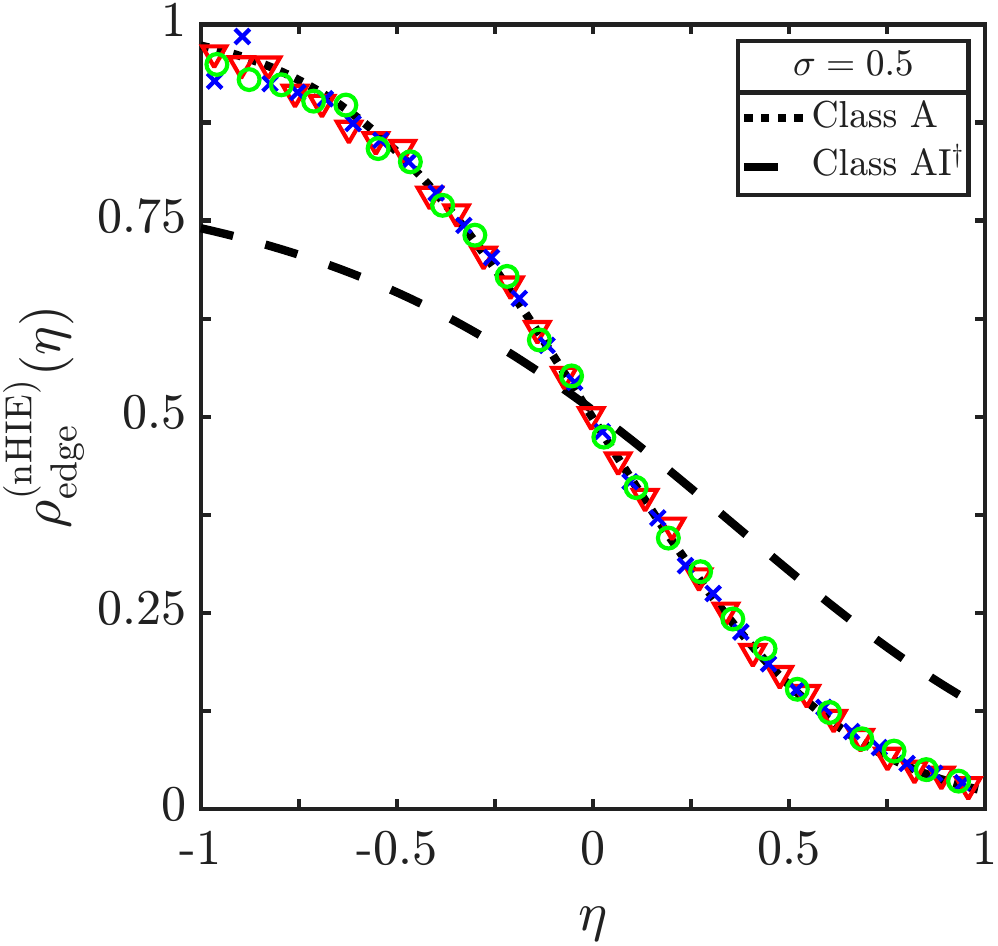}
    \hspace{0.5cm}
    \includegraphics[width=0.48\linewidth]{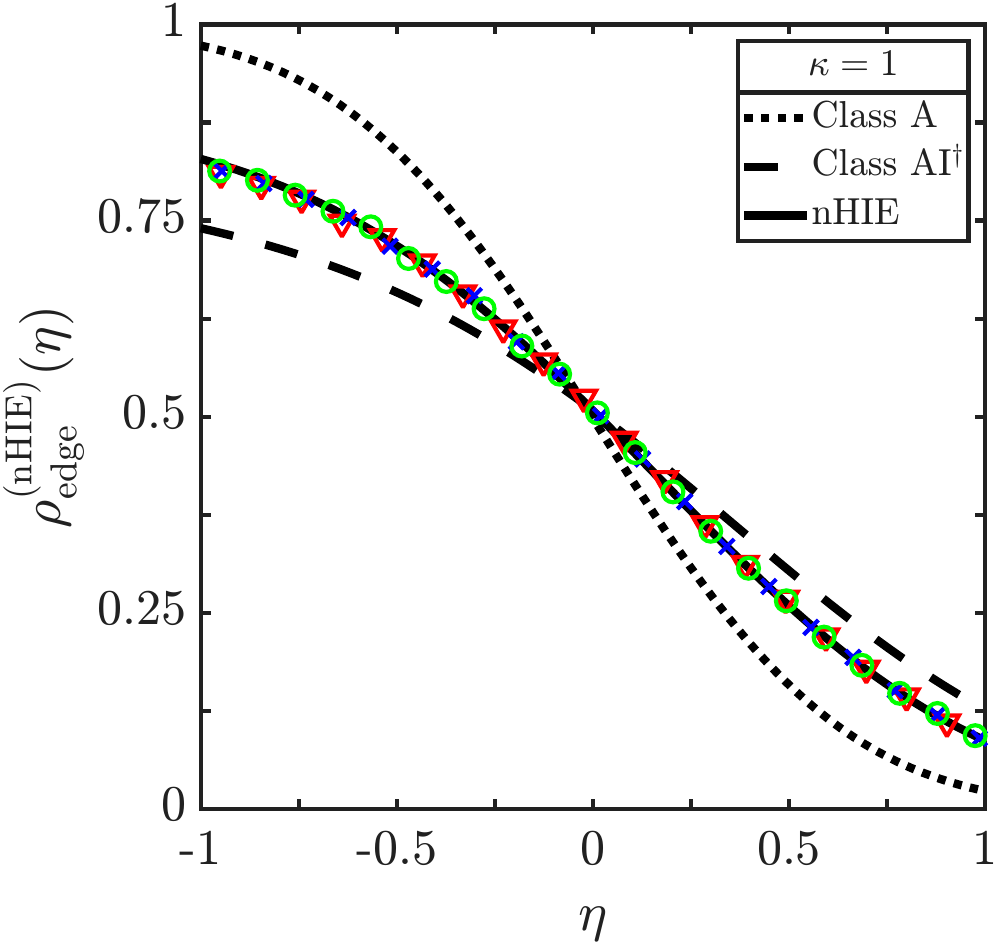}
    \caption{Density of eigenvalues at the edge of the nHIE in the strong (left) and weak (right) asymmetry regimes, with fixed parameters $\sigma=0.5$ and $\kappa=1$ respectively. The dotted (dashed) line corresponds to the density of eigenvalues in the complex Ginibre (symmetric) ensemble, whereas the solid line in the right hand plot is the interpolating edge density from Eq. \eqref{eq:rho_WA_edge}. Note that each density is normalised so as to have its entire mass on the scale shown of $\eta \in [-1,1]$. The coloured markers represent results from numerical simulation of matrices $X_N$ in Eq. \eqref{eq:X_def}, where the entries of $G_N$ are drawn from: Gaussian (red triangles), Bernoulli (blue crosses) and uniform (green circles) distributions. Numerical data points are obtained from approximately $10^5$ samples of $N=2500$ (left) and $N=1000$ (right) matrices. }
    \label{fig:edge_densities}
\end{figure}

\section{Proof of Finite $N$ Results}
\label{sec:finite_proofs}

In this Section, we shall prove our finite matrix size results for the JPDF of eigenvalue and eigenvector and the density of complex eigenvalues, contained in Theorem \ref{thm:JPDF_z_v} and Corollary \ref{cor:rho_finite_N} respectively. To derive this JPDF, we appeal to the \textit{Kac-Rice} method for non-Hermitian matrices, recently pioneered by Fyodorov in \cite{YanKacRice}. Specifically, we utilise \cite[Eq. (3.1)]{YanKacRice}, which shows that one can calculate the JPDF of eigenvalue and normalised eigenvector at finite $N$, via 
\begin{align}
    \mathcal{P}^{(Y)}_{N}(z, \bm v) &= \frac{1}{\pi} \lim_{|w-z|\to 0} \frac{\partial^2}{\partial w \partial \overline{w} } \lim_{\epsilon \to 0} \int_{\mathbb{C}^N} \frac{\diff \bm k \, \diff \bm k^\dagger}{(2 \pi)^{2N}} \, e^{ - \frac{\epsilon \bm k^\dagger \bm k}{2} - \frac{i}{2} [z \bm k^\dagger \bm v + \overline{z} \bm v^\dagger \bm k]} \nonumber\\
    & \times \int D(\bm \Psi, \bm \Phi) \, e^{-\frac{i}{2}\left( w \bm \psi_1^T \bm \varphi_2 + \overline{w} \bm \psi_2^T \bm \varphi_1 \right)} \left \langle e^{\frac{i}{2} \text{Tr}(Y A + Y^{\dagger}B)} \right \rangle_{Y} \ ,
    \label{eq:initial_setup}
\end{align}
where $A \equiv \bm v \bm k^\dagger - \bm \varphi_2 \bm \psi_1^T$ and $B \equiv \bm k \bm v^\dagger - \bm \varphi_1 \bm \psi_2^T$ for a general random matrix ensemble $Y$. Furthermore, the vectors $\bm \varphi_1$, $\bm \varphi_2$, $\bm \psi_1$ and $\bm \psi_2$ contain $N$ Grassmann valued numbers each, indexed as $\varphi_{1i}$ etc., and $D(\bm \Psi, \bm \Phi) \equiv \diff \bm \psi_1 \, \diff \bm \varphi_1 \, \diff \bm \psi_2 \, \diff \bm \varphi_2 $. Grassmann numbers are anti-commuting variables, i.e. $\phi \chi = - \chi \phi$, which are utilised extensively throughout RMT and related mathematical physics problems. For more information on their mathematical definitions and uses in calculations one can consult \cite{Berezin87,Verbaarschot04}.

Note as well that we constantly enforce the normalisation of the eigenvectors through the use of $\delta(\bm v^\dagger \bm v - 1)$, however for the sake of brevity we will suppress this notation. In what follows, we will use the notation $\mathcal{P}^{(\text{nHIE})}_{N}(z, \bm v; \sigma)$ to denote the JPDF in the nHIE, which corresponds to evaluating Eq. \eqref{eq:initial_setup} with $Y=X_N$, as defined in Eq. \eqref{eq:X_def}. Considering Eq. \eqref{eq:initial_setup}, we see that the first object that one must calculate in this endeavour is the expectation of an exponentiated trace. This is done rather easily through the following proposition, which is proved in Appendix  \ref{app:Ensemble_average_identity}.

\begin{prop}
    \label{prop:ensemble_ave}
    For the ensemble of random matrices $X_N$ denoting the nHIE, defined in Eq. \eqref{eq:X_def}, and for fixed $N \times N$ matrices $A$ and $B$, we have that
    \begin{align}
        \left \langle \exp\Big[\Tr(X_N A + X_N^{\dagger}B)\Big] \right \rangle_{X_N} &= \exp\Big[ \Tr(AB) + \sigma \Tr(A^T B) \Big]  \ .  
        \label{eq:ave_expo_trace_nHIE}
    \end{align}
\end{prop}
The form of the resulting exponential in Proposition \ref{prop:ensemble_ave}, as well as the necessary forms of our $A$ and $B$, means that we have to calculate the following traces
\begin{align}
    \Tr(AB) &= \bm k^\dagger \bm k - (\bm v^\dagger \bm \varphi_2) (\bm \psi_1^T \bm k) - (\bm k^\dagger \bm \varphi_1) (\bm \psi_2^T \bm v) - (\bm \psi_1^T \bm \varphi_1) (\bm \psi_2^T \bm \varphi_2) \\
    \Tr(A^T B) &= (\bm k^T \bm v ) \overline{(\bm k^T \bm v )} + (\bm k^\dagger \bm \psi_2) (\bm \varphi_1^T \bm v) - (\bm k^T \bm \varphi_2) (\bm v^\dagger \bm \psi_1) + (\bm \varphi_1^T \bm \varphi_2) (\bm \psi_1^T \bm \psi_2) \ ,
\end{align}
where in the first equation we have utilised the fact that the eigenvector is normalised, i.e. $\bm v^\dagger \bm v = 1$. Upon performing the ensemble average and evaluating above traces, we can now write that:
\begingroup\allowdisplaybreaks
\begin{align}
    \mathcal{P}^{(\text{nHIE})}_{N}(z, \bm v; \sigma) &= \frac{1}{\pi} \lim_{|w-z|\to 0,\epsilon \to 0} \frac{\partial^2}{\partial w \partial \overline{w} } \int D(\bm \Psi, \bm \Phi) \, \exp[-\frac{i\left( w \bm \psi_1^T \bm \varphi_2 + \overline{w} \bm \psi_2^T \bm \varphi_1 \right)}{2} + \frac{(\bm \psi_1^T \bm \varphi_1) (\bm \psi_2^T \bm \varphi_2) - \sigma(\bm \varphi_1^T \bm \varphi_2) (\bm \psi_1^T \bm \psi_2)}{4}  ] \nonumber\\
    & \times \int_{\mathbb{C}^N} \frac{\diff \bm k \, \diff \bm k^\dagger}{(2 \pi)^{2N}}  \exp\Bigg[- \left( \frac{1}{2} + \epsilon \right) \frac{ \bm k^\dagger \bm k}{2} - \frac{i}{2} \big[ z \bm k^\dagger \bm v + \overline{z} \bm v^\dagger \bm k \big] - \frac{\sigma}{4} (\bm k^T \bm v ) \overline{(\bm k^T \bm v )} + \frac{(\bm v^\dagger \bm \varphi_2) (\bm \psi_1^T \bm k)}{4} \nonumber \\
    & + \frac{(\bm k^\dagger \bm \varphi_1) (\bm \psi_2^T \bm v)}{4}  - \frac{\sigma}{4} \Big( (\bm k^\dagger \bm \psi_2) (\bm \varphi_1^T \bm v) - (\bm k^T \bm \varphi_2) (\bm v^\dagger \bm \psi_1) \Big) \Bigg] \ .
\end{align}
\endgroup
We now proceed to identify the overall coefficients of the vectors $\bm k$ and $\bm k^\dagger$, so that we can write the JPDF in the form 
\begingroup\allowdisplaybreaks
\begin{align}
    \mathcal{P}^{(\text{nHIE})}_{N}(z, \bm v; \sigma) &= \frac{1}{\pi} \lim_{|w-z|\to 0,\epsilon \to 0} \frac{\partial^2}{\partial w \partial \overline{w} } \int D(\bm \Psi, \bm \Phi) \, \exp[-\frac{i\left( w \bm \psi_1^T \bm \varphi_2 + \overline{w} \bm \psi_2^T \bm \varphi_1 \right)}{2} + \frac{(\bm \psi_1^T \bm \varphi_1) (\bm \psi_2^T \bm \varphi_2) - \sigma(\bm \varphi_1^T \bm \varphi_2) (\bm \psi_1^T \bm \psi_2)}{4} ] \nonumber\\
    & \times \int_{\mathbb{C}^N} \frac{\diff \bm k \, \diff \bm k^\dagger}{(2 \pi)^{2N}} \, \exp\Bigg[- \left( \frac{1}{2} + \epsilon \right) \frac{ \bm k^\dagger \bm k}{2}  - \frac{\sigma}{4} (\bm k^T \bm v ) \overline{(\bm k^T \bm v )} + \frac{1}{4} \left( \bm k^\dagger \bm b_1 + \bm b_2^T \bm k \right) \Bigg] \ ,
    \label{eq:JPD_pre_k_ints}
\end{align}
\endgroup
where we have introduced
\begin{align}
    \bm b_1 \equiv - 2i z \bm v + \bm \varphi_1 (\bm \psi_2^T \bm v) -  \sigma \bm \psi_2 ( \bm \varphi_1^T \bm v ) \hspace{1cm} \text{and} \hspace{1cm} \bm b_2^T \equiv - 2i \overline{z} \bm v^\dagger + (\bm v^\dagger \bm \varphi_2) \bm \psi_1^T - \sigma (\bm v^\dagger \bm \psi_1) \bm \varphi_2^T \ .
\end{align}
One can now make the change of integration measure $\diff \bm k \diff \bm k^\dagger = 2^N \diff^2 \bm k = 2^N \prod_i \diff \Re(k_i) \diff \Im(k_i)$ and rewrite $(\bm k^T \bm v ) \overline{(\bm k^T \bm v )} = \bm k^\dagger \overline{\bm v} \bm v^T \bm k$, so as to see that
\begin{align}
    \int_{\mathbb{C}^N} \frac{\diff \bm k \, \diff \bm k^\dagger}{(2 \pi)^{2N}} \, e^{- \left( \frac{1}{2} + \epsilon \right) \frac{ \bm k^\dagger \bm k}{2}  - \frac{\sigma}{4} (\bm k^T \bm v ) \overline{(\bm k^T \bm v )} + \frac{1}{4} \left( \bm k^\dagger \bm b_1 + \bm b_2^T \bm k \right) } = \frac{1}{(2 \pi)^N} \int_{\mathbb{C}^N} \frac{\diff^2 \bm k }{\pi^{N}} \, e^{- \bm  k^\dagger C_{\sigma, \epsilon}  \bm k + \frac{1}{4} \left( \bm k^\dagger \bm b_1 + \bm b_2^T \bm k \right) }
    \label{eq:the_Gaussian_int}
\end{align}
where we have introduced
\begin{equation}
    C_{\sigma,\epsilon} \equiv \frac{1}{4} \1_N + \frac{\epsilon}{2} \1_N +  \frac{\sigma}{4} \overline{\bm v} \bm v^T \ .
    \label{eq:C_rho_def}
\end{equation}
Simple analysis yields that $C_{\sigma,\epsilon}$ is positive definite for all $\epsilon \geq 0$ and so we can trivially take the limit of $\epsilon \to 0$. One can also apply Sylvester's determinant theorem, i.e.
\begin{equation}
    \det( \1_N + a \bm v \bm v^\dagger) = 1+a
    \label{eq:Sylvester_det_theorem}
\end{equation}
for $a \in \mathbb{C}$ and $\bm v^\dagger \bm v = 1$, as well as the Sherman-Morrison formula \cite{Sherman1950} to see that
\begin{equation}
    \det(C_{\sigma,0}) = \frac{1}{4^N} (1 + \sigma) \hspace{1cm} \text{and} \hspace{1cm} C_{\sigma,0}^{-1} = 4 \, \1_N - 4 \frac{\sigma}{1 + \sigma} \overline{\bm v} \bm v^T \ . 
    \label{eq:C_det_inv}
\end{equation}
Therefore, now employing the well-known Gaussian identity of  
\begin{equation}
    \frac{1}{\pi^N} \int_{\mathbb{C}^N} \diff^2 \bm x \, e^{-\bm x^\dagger A \bm x - \bm x^\dagger \bm b - \bm c^\dagger \bm x} = \frac{1}{\det(A)} \exp[ {\bm c}^\dagger A^{-1} {\bm b} ] \ ,
    \label{eq:cmplx_Gaussian_identity}
\end{equation}
for positive definite Hermitian matrices $A$ and fixed complex vectors $\bm b$ and $\bm c$, we can see that 
\begin{equation}
    \lim_{\epsilon \to 0} \int_{\mathbb{C}^N} \frac{\diff \bm k \, \diff \bm k^\dagger}{(2 \pi)^{2N}} \, e^{- \frac{ \bm k^\dagger \bm k}{4}  - \frac{\sigma}{4} (\bm k^T \bm v ) \overline{(\bm k^T \bm v )} + \frac{1}{4} \left( \bm k^\dagger \bm b_1 + \bm b_2^T \bm k \right) } = \frac{2^N}{\pi^N (1 + \sigma)} \exp\left[ \frac{1}{4} \bigg( \bm b_2^T \bm b_1  - \mu \, \bm b_1^T \bm v \bm v^\dagger \bm b_2 \bigg) \right] \ ,
\end{equation}
where we have introduced $\mu \equiv \sigma/(1 + \sigma)$ and employed $\bm b_2^T \overline{\bm v} \bm v^T \bm b_1 = \bm b_1^T \bm v \bm v^\dagger \bm b_2$, since $\bm b_1$ and $\bm b_2$ are both even Grassmann variables (i.e. each term is either a constant or a product with an even number of Grassmann variables). Hence, we must now write down the following formulas
\begingroup
\allowdisplaybreaks
\begin{align}
    \bm b_2^T \bm b_1 &= - 4|z|^2 + 2 i z (\bm \psi_1^T \bm v) (\bm v^\dagger \bm \varphi_2) + 2 i \overline{z} (\bm \psi_2^T \bm v) (\bm v^\dagger \bm \varphi_1) - (\bm \psi_1^T \bm \varphi_1) (\bm \psi_2^T \bm v) (\bm v^\dagger \bm \varphi_2) - \sigma^2 (\bm \varphi_2^T \bm \psi_2) (\bm \varphi_1^T \bm v) (\bm v^\dagger \bm \psi_1) \nonumber \\
    & + \sigma \Big[ - 2 i z (\bm \varphi_2^T \bm v) (\bm v^\dagger \bm \psi_1) - 2 i \overline{z} (\bm \varphi_1^T \bm v) (\bm v^\dagger \bm \psi_2) + (\bm \varphi_2^T \bm \varphi_1) (\bm \psi_2^T \bm v) (\bm v^\dagger \bm \psi_1) + (\bm \psi_1^T \bm \psi_2) (\bm \varphi_1^T \bm v) (\bm v^\dagger \bm \varphi_2) \Big]
\end{align}
\endgroup
and, introducing $\nu \equiv \bm v^T \bm v$,
\begingroup\allowdisplaybreaks
\begin{align}
    &\frac{\bm b_1^T \bm v \bm v^\dagger \bm b_2}{4} = \frac{1}{4} \Big( - 2 i \overline{z} \, \overline{\nu} - (1+ \sigma) (\bm \psi_1^T \overline{\bm v}) (\bm v^\dagger \bm \varphi_2) \Big) \Big(  - 2i z \nu - (1 + \sigma) (\bm \psi_2^T \bm v) ( \bm v^T  \bm \varphi_1)  \Big) \nonumber\\
    &= - |z|^2 |\nu|^2 + \frac{i (1 + \sigma )}{2}  \big[  z  \nu (\bm \psi_1^T \overline{\bv}) (\bv^\dagger \bm \varphi_2) + \overline{z} \, \overline{\nu} (\bm \psi_2^T \bv ) (\bm v^T \bm \varphi_1)\big] + \frac{( 1 + \sigma )^2}{4} (\bm \psi_1^T \overline{\bm v}) (\bm v^\dagger \bm \varphi_2) (\bm \psi_2^T \bm v) ( \bm v^T  \bm \varphi_1)  \ ,
\end{align}
\endgroup
so that we can eventually return to evaluating the necessary Grassmannian integrals in Eq. \eqref{eq:JPD_pre_k_ints}.  Note at this point that when given the choice of how to express an inner product, we always write them in the form $\bm \psi_i^T \bm \varphi_j$ or $\bm \psi_1^T \bm \psi_2$/ $\bm \varphi_1^T \bm \varphi_2$ - the reason for this should become clear when constructing the Pfaffian. Therefore, the exponential that results from the Gaussian integral can be written as
\begin{align}
    \exp\bigg[&\bm b_2^T C_{\sigma,0}^{-1} \bm b_1\bigg] = e^{- (1 - \mu|\nu|^2) |z|^2}  \exp\Bigg[ \frac{i}{2} \bigg( z (\bm \psi_1^T \bm v) (\bm v^\dagger \bm \varphi_2) + \overline{z} (\bm \psi_2^T \bm v) (\bm v^\dagger \bm \varphi_1) + z \sigma (\bm \psi_1^T \overline{\bv}) (\bv^T \bm \varphi_2) + \overline{z} \sigma (\bm \psi_2^T \overline{\bv} ) (\bv^T \bm \varphi_1)   \bigg) \nonumber \\
    &- \frac{i \mu (1 + \sigma )}{2} \Big[  z  \nu (\bm \psi_1^T \overline{\bv}) (\bv^\dagger \bm \varphi_2) + \overline{z} \, \overline{\nu} (\bm \psi_2^T \bv ) (\bm v^T \bm \varphi_1 )\Big] + \frac{1}{4} \bigg(- (\bm \psi_1^T \bm \varphi_1) (\bm \psi_2^T \bm v) (\bm v^\dagger \bm \varphi_2) + \sigma (\bm \varphi_1^T \bm \varphi_2) (\bm \psi_1^T \overline{\bv}) (\bv^T\bm \psi_2) \nonumber \\
    & + \sigma (\bm \psi_1^T \bm \psi_2) (\bm \varphi_1^T \bm v) (\bm v^\dagger \bm \varphi_2) - \sigma^2 (\bm \psi_2^T \bm \varphi_2) (\bm \psi_1^T \overline{\bv}) (\bv^T\bm \varphi_1) - \mu ( 1 + \sigma )^2 (\bm \psi_1^T \overline{\bm v}) (\bm v^\dagger \bm \varphi_2) (\bm \psi_2^T \bm v) ( \bm v^T  \bm \varphi_1)\bigg) \Bigg] 
\end{align}
then, inserting the above equation into Eq. \eqref{eq:JPD_pre_k_ints} and combining some of the terms which are quartic in Grassmannian variables, we see that
\begingroup\allowdisplaybreaks
\begin{align}
     \mathcal{P}^{(\text{nHIE})}_{N}(z, \bm v; \sigma) &= \frac{2^N e^{-(1 - \mu|\nu|^2)|z|^2}}{\pi^{N+1} (1 + \sigma)}  \lim_{|w-z|\to 0} \frac{\partial^2}{\partial w \partial \overline{w} } \int D(\bm \Psi, \bm \Phi) \exp\Bigg[ - \frac{i \mu (1 + \sigma )}{2} \Big[  z  \nu (\bm \psi_1^T \overline{\bv}) (\bv^\dagger \bm \varphi_2) + \overline{z} \, \overline{\nu} (\bm \psi_2^T \bv ) (\bm v^T \bm \varphi_1 )\Big]  \nonumber \\ 
    &+\frac{i}{2} \bigg( z (\bm \psi_1^T \bm v) (\bm v^\dagger \bm \varphi_2) + \overline{z} (\bm \psi_2^T \bm v) (\bm v^\dagger \bm \varphi_1) + z \sigma (\bm \psi_1^T \overline{\bv}) (\bv^T \bm \varphi_2) + \overline{z} \sigma (\bm \psi_2^T \overline{\bv} ) (\bv^T \bm \varphi_1) - w \bm \psi_1^T \bm \varphi_2 - \overline{w} \bm \psi_2^T \bm \varphi_1  \bigg) \nonumber \\
    & + \frac{1}{4} \bigg( (\bm \psi_1^T \bm \varphi_1) \, \bm \psi_2^T  (\1_N - \bm v \bm v^\dagger) \bm \varphi_2 + \sigma (\bm \varphi_1^T \bm \varphi_2) \, \bm \psi_1^T  ( \overline{\bv} \bv^T- \1_N) \bm \psi_2  + \sigma (\bm \psi_1^T \bm \psi_2) (\bm \varphi_1^T \bm v) (\bm v^\dagger \bm \varphi_2) \nonumber \\    
    &  - \sigma^2 (\bm \psi_2^T \bm \varphi_2) (\bm \psi_1^T \overline{\bv}) (\bv^T\bm \varphi_1) - \mu ( 1 + \sigma )^2 (\bm \psi_1^T \overline{\bm v}) (\bm v^\dagger \bm \varphi_2) (\bm \psi_2^T \bm v) ( \bm v^T  \bm \varphi_1)\bigg) \Bigg] \ .
\end{align}
\endgroup
Gathering the quartic terms in this manner necessitates the need to introduce five Hubbard-Stratonovich transformations, specifically:
\begingroup\allowdisplaybreaks
\begin{align}
    &\exp\bigg[\frac{(\bm \psi_1^T \bm \varphi_1)}{4} \bm \psi_2^T \Big[ \1_N - \bm v \bm v^\dagger \Big] \bm \varphi_2\bigg] =  \int_{\mathbb{C}} \frac{ \diff^2 q_1}{\pi}  e^{-|q_1|^2} \exp\bigg[ - \frac{q_1}{2} \Big(\bm \psi_1^T \bm \varphi_1 \Big) - \frac{\overline{q_1}}{2} \Big(\bm \psi_2^T \Big[ \1_N - \bm v \bm v^\dagger \Big] \bm \varphi_2\Big) \bigg]\\
    &\exp\bigg[\frac{\sigma\Big(\bm \varphi_1^T \bm \varphi_2\Big)}{4} \, \bm \psi_1^T  \Big[ \overline{\bv} \bv^T- \1_N\Big] \bm \psi_2 \bigg] = \int_{\mathbb{C}} \frac{ \diff^2 q_2}{\pi}  e^{-|q_2|^2} \exp\bigg[- \frac{\sigma q_2}{2} \Big( \bm \varphi_1^T \bm \varphi_2 \Big) - \frac{\overline{q_2}}{2}\Big( \bm \psi_1^T \Big[ \overline{\bv} \bv^T- \1_N \Big] \bm \psi_2 \Big) \bigg]\\
    &\exp\bigg[\frac{\sigma}{4} (\bm \psi_1^T \bm \psi_2) (\bm \varphi_1^T \bm v) (\bm v^\dagger \bm \varphi_2)\bigg] = \int_{\mathbb{C}} \frac{ \diff^2 q_3}{\pi}  e^{-|q_3|^2}\exp\bigg[ - \frac{\sigma q_3}{2}\Big( \bm \psi_1^T \bm \psi_2 \Big) - \frac{ \overline{q_3}}{2}\Big( \bm \varphi_1^T \bm v \bm v^\dagger \bm \varphi_2 \Big) \bigg]\\
    &\exp\bigg[ -\frac{\sigma^2}{4} (\bm \psi_2^T \bm \varphi_2) (\bm \psi_1^T \overline{\bv}) (\bm v^T \bm \varphi_1) \bigg] = \int_{\mathbb{C}} \frac{ \diff^2 q_4}{\pi}  e^{-|q_4|^2} \exp\bigg[ - \frac{i\sigma q_4}{2}\Big( \bm \psi_2^T \bm \varphi_2 \Big) - \frac{i\sigma\overline{q_4}}{2}\Big( \bm \psi_1^T \overline{\bv} \bm v^T \bm \varphi_1 \Big) \bigg] \\
    &\exp[- \frac{\mu ( 1 + \sigma )^2}{4}  (\bm \psi_1^T \overline{\bm v}) (\bm v^\dagger \bm \varphi_2) (\bm \psi_2^T \bm v) ( \bm v^T  \bm \varphi_1)] = \int_{\mathbb{C}} \frac{ \diff^2 q_5}{\pi}  e^{-|q_5|^2}  \exp[ - \frac{i \mu q_5}{2} \Big( \bm \psi_1^T \overline{\bm v} \bm v^\dagger \bm \varphi_2 \Big) - \frac{i (1 + \sigma)^2 \overline{q_5}}{2} \Big(\bm \psi_2^T \bm v \bm v^T  \bm \varphi_1 \Big)  ] \ ,
\end{align}
\endgroup
where we again use the convention that $\diff ^2 q_i = \diff \Re(q_i) \diff \Im(q_i)$. Note that the transformations are chosen in this manner so as to ensure we have no square roots to deal with in the ensuing expansion of the Pfaffian. One is of course free to choose these transformations in any valid way, but the form of the resulting Pfaffian will change for a different choice. The final result will however be independent of this choice. We now incorporate these transformations into our equation for $\mathcal{P}^{(\text{nHIE})}_{N}(z, \bm v; \sigma)$ to arrive at
\begin{align}
    \mathcal{P}^{(\text{nHIE})}_{N}(z, \bm v; \sigma) &= \frac{2^N e^{-(1 - \mu|\nu|^2)|z|^2}}{\pi^{N+6} (1 + \sigma)} \int_{\mathbb{C}^5} \diff^2 q_1 \, e^{-|q_1|^2} \cdots \diff^2 q_5 \, e^{-|q_5|^2} \lim_{|w-z|\to 0} \frac{\partial^2}{\partial w \partial \overline{w} } \int D(\bm \Psi, \bm \Phi) \exp\Bigg[ - \frac{i}{2} \bigg(  w \bm \psi_1^T \bm \varphi_2 + \overline{w} \bm \psi_2^T \bm \varphi_1 \bigg) \nonumber \\
    & + \frac{i}{2} \bigg( z (\bm \psi_1^T \bm v) (\bm v^\dagger \bm \varphi_2) + \overline{z} (\bm \psi_2^T \bm v) (\bm v^\dagger \bm \varphi_1) + z \sigma (\bm \psi_1^T \overline{\bv}) (\bv^T \bm \varphi_2) + \overline{z} \sigma (\bm \psi_2^T \overline{\bv} ) (\bv^T \bm \varphi_1)  \bigg)  - \frac{q_1}{2} \Big(\bm \psi_1^T \bm \varphi_1 \Big) - \frac{\sigma q_2}{2} \Big( \bm \varphi_1^T \bm \varphi_2 \Big) \nonumber \\ 
    &  - \frac{\sigma q_3}{2}\Big( \bm \psi_1^T \bm \psi_2 \Big) - \frac{i\sigma q_4}{2}\Big( \bm \psi_2^T \bm \varphi_2 \Big) - \frac{i \mu (1 + \sigma )}{2} \Big[  z  \nu (\bm \psi_1^T \overline{\bv}) (\bv^\dagger \bm \varphi_2) + \overline{z} \, \overline{\nu} (\bm \psi_2^T \bv ) (\bm v^T \bm \varphi_1 )\Big] - \frac{\overline{q_1}}{2} \Big(\bm \psi_2^T \Big[ \1_N - \bm v \bm v^\dagger \Big] \bm \varphi_2\Big) \nonumber \\
    & - \frac{\overline{q_2}}{2}\Big( \bm \psi_1^T \Big[ \overline{\bv} \bv^T- \1_N \Big] \bm \psi_2 \Big) - \frac{ \overline{q_3}}{2}\Big( \bm \varphi_1^T \bm v \bm v^\dagger \bm \varphi_2 \Big) - \frac{i\sigma\overline{q_4}}{2}\Big( \bm \psi_1^T \overline{\bv} \bm v^T \bm \varphi_1 \Big) \nonumber \\
    & - \frac{i \mu q_5}{2} \Big( \bm \psi_1^T \overline{\bm v} \bm v^\dagger \bm \varphi_2 \Big) - \frac{i (1 + \sigma)^2 \overline{q_5}}{2} \Big(\bm \psi_2^T \bm v \bm v^T  \bm \varphi_1 \Big)  \Bigg] \ ,
\end{align}
and so we are now at a stage where we can write the exponential in the form $- \frac{1}{2} \bm \chi^T \mathcal{M} \bm \chi$, where $\bm \chi^T = (\bm \psi_1^T, \bm \psi_2^T, \bm \varphi_1^T, \bm \varphi_2^T)$. The reason for wanting to do this is to facilitate the use of
\begin{equation}
    \int D \bm \chi \exp[- \frac{1}{2} \bm \chi^T \mathcal{M} \bm \chi] = \Pf(\mathcal{M})\ ,
\end{equation}
where, for us, $D \bm \chi = (-1)^N D(\bm \Psi, \bm \Phi)$ and $\Pf(\mathcal{M})$ denotes the Pfaffian of $\mathcal{M}$ \cite{Berezin87}. This analysis is most easily supported for $\mathcal{M}$ being anti-symmetric and so to this end we appeal to the following useful trick for Grassmann numbers. Namely, for Grassmann numbers $\alpha$ and $\beta$, we can always say that $\alpha \beta = 0.5( \alpha \beta - \beta \alpha)$, this means that for an arbitrary matrix $M$ we can appeal to the general relation
\begin{equation}
    \bm \psi^T M \bm \varphi = \frac{1}{2} \bigg( \bm \psi^T M \bm \varphi  - \bm \varphi^T M^T \bm \psi  \bigg) \ .
    \label{eq:Grass_argument}
\end{equation}
To this end, we aim to identify the upper right-hand block of $\mathcal{M}$ with the following contributions
\begingroup
\allowdisplaybreaks
\begin{align}
    &- \frac{1}{2} \bm \psi_1^T \mathcal{K}_{11} \, \bm \varphi_1 \hspace{1cm} \mathcal{K}_{11} \equiv q_1 \1_N + i \sigma \overline{q_4} \, \overline{\bv} \bm v^T \label{eq:K_11} \\
    &- \frac{1}{2} \bm \psi_2^T \mathcal{K}_{22} \, \bm \varphi_2 \hspace{1cm} \mathcal{K}_{22} \equiv ( \overline{q_1} + i \sigma q_4) \1_N - \overline{q_1} \bm v \bm v^\dagger \\
    &- \frac{i}{2} \bm \psi_1^T \mathcal{K}_{12} \, \bm \varphi_2 \hspace{1cm} \mathcal{K}_{12} \equiv w \1_N - z \bm v \bm v^\dagger - \sigma z \overline{\bm v} \bm v^T + \left[ \mu (1+ \sigma ) z \nu + \mu q_5 \right] \overline{\bm v} \bm v^\dagger \\
    &- \frac{i}{2} \bm \psi_2^T \mathcal{K}_{21} \, \bm \varphi_1 \hspace{1cm} \mathcal{K}_{21} \equiv \overline{w} \1_N - \overline{z} \bm v \bm v^\dagger - \sigma \overline{z} \, \overline{\bm v} \bm v^T +  \left[ \mu (1+ \sigma) \overline{z} \, \overline{\nu} + (1+ \sigma)^2 \overline{q_5} \right] \bm v \bm v^T 
\end{align}
\endgroup
and then from the terms linking $\bm \varphi_1^T \bm \varphi_2$ and $\bm \psi_1^T \bm \psi_2$ we have
\begingroup
\allowdisplaybreaks
\begin{align}
    &- \frac{1}{2} \bm \varphi_1^T \mathcal{S}_{\varphi} \, \bm \varphi_2 \hspace{1cm} \mathcal{S}_{\varphi} \equiv \sigma q_2 \1_N + \overline{q_3} \bm v \bm v^\dagger \\
    &- \frac{1}{2} \bm \psi_1^T \mathcal{S}_{\psi} \, \bm \psi_2 \hspace{1cm} \mathcal{S}_{\psi} \equiv ( \sigma q_3 - \overline{q_2})  \1_N + \overline{q_2} \, \overline{\bv} \bm v^T \label{eq:S_psi_def}
\end{align}
\endgroup
note the choice of $\mathcal{S}$ for `self'. With these matrices in place, we can apply the logic of Eq. \eqref{eq:Grass_argument} to see that
\begingroup\allowdisplaybreaks
\begin{align}
   \mathcal{P}&^{(\text{nHIE})}_{N}(z, \bm v; \sigma) = \frac{2^N e^{-(1 - \mu|\nu|^2)|z|^2}}{\pi^{N+6} (1 + \sigma)} \int_{\mathbb{C}^5} \diff^2 q_1 \, e^{ - |q_1|^2} \cdots \diff^2 q_5  \, e^{ - |q_5|^2} \lim_{|w-z|\to 0}   \frac{\partial^2}{\partial w \partial \overline{w} } \nonumber\\
    & \times \int D(\bm \Psi, \bm \Phi) \exp[ - \frac{1}{4} 
    \begin{matrix}
        \Big(\bm \psi_1^T, & \bm \psi_2^T, & \bm \varphi_1^T, & \bm \varphi_2^T \Big) \\ \\ \\ \\
    \end{matrix} 
    \left( \begin{matrix}
        0 & \mathcal{S}_\psi & \mathcal{K}_{11} & i\mathcal{K}_{12} \\
        - \mathcal{S}_\psi^T & 0 & i\mathcal{K}_{21} & \mathcal{K}_{22} \\
        - \mathcal{K}_{11}^T & - i\mathcal{K}_{21}^T & 0 & \mathcal{S}_\varphi \\
        - i\mathcal{K}_{12}^T & - \mathcal{K}_{22}^T & - \mathcal{S}_\varphi^T & 0
    \end{matrix}  \right)
    \left( \begin{matrix}
        \bm \psi_1 \\ \bm \psi_2 \\ \bm \varphi_1 \\ \bm \varphi_2 
    \end{matrix}  \right) ] \nonumber \\
    &= \frac{e^{-(1 - \mu|\nu|^2)|z|^2}}{\pi^{N+6} (1 + \sigma)} \int_{\mathbb{C}^5} \diff^2 q_1 \, e^{ - |q_1|^2} \cdots \diff^2 q_5  \, e^{ - |q_5|^2} \lim_{|w-z|\to 0}   \frac{\partial^2}{\partial w \partial \overline{w} } (-1)^N \Pf\left( \begin{matrix}
        0 & \mathcal{S}_\psi & \mathcal{K}_{11} & i\mathcal{K}_{12} \\
        - \mathcal{S}_\psi^T & 0 & i\mathcal{K}_{21} & \mathcal{K}_{22} \\
        - \mathcal{K}_{11}^T & - i\mathcal{K}_{21}^T & 0 & \mathcal{S}_\varphi \\
        - i\mathcal{K}_{12}^T & - \mathcal{K}_{22}^T & - \mathcal{S}_\varphi^T & 0
    \end{matrix}  \right) \, .
    \label{eq:JPDF_with_Pfaff}
\end{align}
\endgroup
The task that remains at this point is to identify which terms in the resulting Pfaffian will survive the combined action of the derivative and the consequent limit being taken. To this end, we denote our matrix of interest as
\begin{equation}
    \mathcal{M} \equiv \left( \begin{matrix}
        0 & \mathcal{S}_\psi & \mathcal{K}_{11} & i\mathcal{K}_{12} \\
        - \mathcal{S}_\psi^T & 0 & i\mathcal{K}_{21} & \mathcal{K}_{22} \\
        - \mathcal{K}_{11}^T & - i\mathcal{K}_{21}^T & 0 & \mathcal{S}_\varphi \\
        - i\mathcal{K}_{12}^T & - \mathcal{K}_{22}^T & - \mathcal{S}_\varphi^T & 0
    \end{matrix} \right) \ ,
    \label{eq:M_def}
\end{equation}
and realise that its Pfaffian is closely related to its determinant, such that
\begin{equation}
    \mathcal{F}^2 \equiv \Pf(\mathcal{M})^2 = \det(\mathcal{M}) \ ,
\end{equation}
where, since $\mathcal{M}$ is anti-symmetric, the determinant should be a perfect square. In what follows, it is hugely convenient to express the $\mathcal{K}$ and $\mathcal{S}$ matrices as follows:
\begin{align}
    \mathcal{S}_\psi = \alpha_\psi \1_N + \RV \Lambda_\psi \RV^T \, , \hspace{0.5cm} 
    &\mathcal{S}_\varphi = \alpha_\varphi \1_N + \RV \Lambda_\varphi \RV^T \, ,  \hspace{0.5cm}
    \mathcal{K}_{11} = \alpha_{11} \1_N + \RV \Lambda_{11} \RV^T \, ,  \hspace{0.5cm}  
    \mathcal{K}_{22} = \alpha_{22} \1_N + \RV \Lambda_{22} \RV^T \nonumber\\
    &\mathcal{K}_{12} = w \1_N + \RV \Lambda_{12} \RV^T \, ,  \hspace{0.5cm} \mathcal{K}_{21} = \overline{w} \1_N + \RV \Lambda_{21} \RV^T  \ ,
    \label{eq:Lambda_matrices_def}
\end{align}
where $\RV = (\bv, \overline{\bv})$ is an $N \times 2$ matrix. Then, through Eqs. \eqref{eq:K_11} - \eqref{eq:S_psi_def} we identify the necessary form of the $\alpha$ and $\beta$ constants as
\begin{align}
    \alpha_{11} = q_1 \, , \hspace{1cm} 
    \alpha_{22} = \overline{q_1} + i \sigma q_4 \, , \hspace{1cm}
    \alpha_{\varphi} = \sigma q_2  \, , \hspace{1cm}
    \alpha_\psi = \sigma q_3 - \overline{q_2} \, , \hspace{1cm}
    \beta_\varphi = \overline{q_3}\, , \hspace{1cm} \beta_\psi = \overline{q_2} \label{eq:our_alphas_betas}
\end{align}
and the $\Lambda$ matrices as
\begin{gather}
    \Lambda_{11} = \left( 
    \begin{matrix}
        0 & 0 \\ i \sigma \overline{q_4} & 0
    \end{matrix}
    \right)\, , \hspace{1cm}
    \Lambda_{22} = \left( 
    \begin{matrix}
        0 & -\overline{q_1} \\ 0 & 0 
    \end{matrix}
    \right)\, \nonumber \\
    \Lambda_{12} = \left( 
    \begin{matrix}
        0 & -z \\ -\sigma z &  \mu ( (1 + \sigma) z \nu + q_5 )
    \end{matrix}
    \right)\, , \hspace{1cm}
    \Lambda_{21} = \left( 
    \begin{matrix}
        (1 + \sigma) ( \mu \overline{z}\, \overline{\nu} + (1 + \sigma) \overline{q_5}) & -\overline{z} \\ -\sigma\overline{z} & 0 
    \end{matrix}
    \right) \ .
    \label{eq:our_Lambdas}
\end{gather}
Through a rather long, and slightly tedious, exercise in linear algebra (outlined in Appendix \ref{app:simplifying_our_Pfaffian}) one can map the Pfaffian of this $4N \times 4N$ matrix to the Pfaffian of a $4 \times 4$ matrix. Specifically,
\begin{equation}
    \mathcal{F}^2 = \left( \frac{\alpha_\varphi + \beta_\varphi}{\alpha_\varphi} \right)^2 \zeta^{2(N-2)} \det[ \zeta \1_4  + \left( \begin{matrix}
    \mathcal{T}_{12} \RV^T \RV & - i \mathcal{T}_{11} \RV^T \RV   \\
    i \mathcal{T}_{22} \RV^T \RV &  \mathcal{T}_{12}^T \RV^T \RV 
    \end{matrix} \right) ] \ ,
\end{equation}
where $\zeta \equiv \alpha_\varphi \alpha_\psi - \alpha_{11} \alpha_{22} - |w|^2$ and $\mathcal{T}_{11}$, $\mathcal{T}_{12}$ and $\mathcal{T}_{22}$ are $2 \times 2$ matrices defined in Appendix \ref{app:simplifying_our_Pfaffian}. Then, we can utilise Mathematica to evaluate the exact form of $\mathcal{F}^2$, which yields a complicated perfect square of the form
\begin{equation}
    \mathcal{F}^2 = \Big(\mathcal{A}_{0,0} + \mathcal{A}_{1,0} \Delta + \mathcal{A}_{0,1} \overline{\Delta} + \mathcal{A}_{1,1} |\Delta|^2 + \cdots \Big)^2
\end{equation}
where $\mathcal{A}_{0,0}$, $\mathcal{A}_{1,0}$, $\mathcal{A}_{0,1}$ and $\mathcal{A}_{1,1}$ are known in terms of the complex $q$ integration variables, $\Delta \equiv (w - z)$ other model parameters. Note that the $\cdots$ represents terms of order higher than $|\Delta|^2$, which do not contribute to the JPDF due to the action of the derivative and the limit - from now on we shall neglect these terms. Thus we can say that
\begin{align}
    \mathcal{P}^{(\text{nHIE})}_{N}&(z, \bm v; \sigma) = \frac{e^{-(1 - \mu|\nu|^2)|z|^2}}{\pi^{N+6} (1 + \sigma)} \int_{\mathbb{C}^5} \diff^2 q_1 \, e^{ - |q_1|^2} \cdots \diff^2 q_5  \, e^{ - |q_5|^2} \lim_{|\Delta|\to 0}   \frac{\partial^2}{\partial \Delta \partial \overline{\Delta} }  \nonumber\\
    & \times \Big( |q_1|^2 + \sigma |q_2|^2 + |\Delta + z|^2 + i \sigma q_1 q_4 - \sigma^2 q_2 q_3  \Big)^{N-2} \Big( \mathcal{A}_{0,0} + \mathcal{A}_{1,0} \Delta + \mathcal{A}_{0,1} \overline{\Delta} + \mathcal{A}_{1,1} |\Delta|^2 + \cdots \Big) \ ,
    \label{eq:JPDF_start_step4}
\end{align} 
and so we realise that the simplest step at this point is to perform the integration over $q_5$, as this does not appear in the bracket which is raised to the power $N-2$. This yields
\begingroup\allowdisplaybreaks
\begin{align}
    \mathcal{P}^{(\text{nHIE})}_{N}(z, \bm v; \sigma) &= \frac{e^{-(1 - \mu|\nu|^2)|z|^2}}{\pi^{N+5} (1 + \sigma)} \int_{\mathbb{C}^4} \diff^2 q_1 \, e^{ - |q_1|^2} \cdots \diff^2 q_4  \, e^{ - |q_4|^2} \lim_{|\Delta|\to 0} \frac{\partial^2}{\partial \Delta \partial \overline{\Delta} }  \nonumber\\
    & \times \Big( |q_1|^2 + \sigma |q_2|^2 + |\Delta + z|^2 + i \sigma q_1 q_4 - \sigma^2 q_2 q_3  \Big)^{N-2} \Big( \mathcal{B}_{0,0} + \overline{z} \, \mathcal{B}_{1,0} \Delta + z \mathcal{B}_{0,1} \overline{\Delta} + \mathcal{B}_{1,1} |\Delta|^2 \Big) \ ,
\end{align} 
\endgroup
where $\mathcal{B}_{0,0}$, $\mathcal{B}_{1,0}$, $\mathcal{B}_{0,1}$ and $\mathcal{B}_{1,1}$ are known polynomials in $q_1, \ldots, q_4$ and stored in Mathematica. At this point it is instructive and useful to note that the integrals over $q_3$ and $q_4$ yield the following
\begingroup\allowdisplaybreaks
\begin{align}
    \int_{\mathbb{C}^2} \diff^2 q_3 \, \diff^2 q_4 \, e^{ - |q_3|^2 - |q_4|^2} \, \Big(  i \sigma q_1 q_4 - \sigma^2 q_2 q_3 \Big)^{\ell} \mathcal{B}_{0,0} = \int_{\mathbb{C}^2} \diff^2 q_3 \, \diff^2 q_4 \, e^{ - |q_3|^2 - |q_4|^2} \, \Big(  i \sigma q_1 q_4 - \sigma^2 q_2 q_3 \Big)^{\ell} \mathcal{B}_{1,0} = 0
\end{align}
\endgroup
for power $\ell = 0,1,2$, with non-zero results when considering $\mathcal{B}_{0,1}$ and $\mathcal{B}_{1,1}$ for $\ell = 0$ and 1 (that we will consider later), which manifests
\begingroup\allowdisplaybreaks
\begin{align}
    \mathcal{P}^{(\text{nHIE})}_{N}(z, \bm v; \sigma) &= \frac{e^{-(1 - \mu|\nu|^2)|z|^2}}{\pi^{N+5} (1 + \sigma)} \int_{\mathbb{C}^4} \diff^2 q_1 \, e^{ - |q_1|^2} \cdots \diff^2 q_4  \, e^{ - |q_4|^2} \lim_{|\Delta|\to 0} \frac{\partial^2}{\partial \Delta \partial \overline{\Delta} }  \nonumber\\
    & \times \Big( |q_1|^2 + \sigma |q_2|^2 + |\Delta + z|^2 + i \sigma q_1 q_4 - \sigma^2 q_2 q_3  \Big)^{N-2} \Big( z \mathcal{B}_{0,1} \overline{\Delta} + \mathcal{B}_{1,1} |\Delta|^2 \Big) \ .
\end{align} 
\endgroup
One now considers the expansion of the term with the exponent $N-2$ so as to ascertain which terms survive the derivative being performed, after some thought this yields
\begingroup\allowdisplaybreaks
\begin{align}
    \Big( |q_1|^2& + \sigma |q_2|^2  + |z|^2 + |\Delta|^2 + z \overline{\Delta} + \overline{z} \Delta + i \sigma q_1 q_4 - \sigma^2 q_2 q_3  \Big)^{N-2} \nonumber \\
    &= \Big( |q_1|^2 + \sigma |q_2|^2 + |z|^2 \Big)^{N-2} +
    (N-2) \Big( |q_1|^2 + \sigma |q_2|^2 + |z|^2 \Big)^{N-3} \Big( |\Delta|^2 + z \overline{\Delta} + \overline{z} \Delta + i \sigma q_1 q_4 - \sigma^2 q_2 q_3  \Big) \nonumber \\
    & + (N-2)(N-3) \Big( |q_1|^2 + \sigma |q_2|^2 + |z|^2 \Big)^{N-4} \Big( |\Delta|^2 + z \overline{\Delta} + \overline{z} \Delta \Big) \Big( i \sigma q_1 q_4 - \sigma^2 q_2 q_3 \Big) + \cdots \ ,
\end{align}
\endgroup
such that, once again, $\cdots$ indicates higher order terms in $|\Delta|$. Then, by identifying the terms that are proportional to $|\Delta|^2$, we can safely take the derivative and the limit to see that
\begingroup\allowdisplaybreaks
\begin{align}
    \mathcal{P}^{(\text{nHIE})}_{N}(z, \bm v; \sigma) &= \frac{e^{-(1 - \mu|\nu|^2)|z|^2}}{\pi^{N+5} (1 + \sigma)} \int_{\mathbb{C}^4} \diff^2 q_1 \, e^{ - |q_1|^2} \cdots \diff^2 q_4  \, e^{ - |q_4|^2} \Bigg[ \Big( |q_1|^2 + \sigma |q_2|^2 + |z|^2 \Big)^{N-2} \mathcal{B}_{1,1} \nonumber\\
    & + (N-2) \Big( |q_1|^2 + \sigma |q_2|^2 + |z|^2 \Big)^{N-3} \Big( |z|^2 \mathcal{B}_{0,1}  + \big(i \sigma q_1 q_4 - \sigma^2 q_2 q_3 \big) \mathcal{B}_{1,1}\Big) \nonumber\\
    & + (N-2)(N-3) |z|^2 \Big( |q_1|^2 + \sigma |q_2|^2 + |z|^2 \Big)^{N-4} \Big( i \sigma q_1 q_4 - \sigma^2 q_2 q_3 \Big) \mathcal{B}_{0,1} \Bigg]\ .
\end{align}
\endgroup
At this point, we are now in a position to perform the $q_3$ and $q_4$ integrals so as to introduce to following functions
\begin{align}
    \mathcal{C}_1 &\equiv \frac{1}{\pi^2} \int_{\mathbb{C}^2} \diff^2 q_3 \, \diff^2 q_4 \, e^{ - |q_3|^2 - |q_4|^2} \, \mathcal{B}_{1,1} = |q_1|^2 + \sigma |q_2|^2 + |z|^2 - \big(2 \sigma |z|^2 + \sigma(1+ \sigma)\big) (1 - |\nu|^2) + |z|^2 \sigma^2 (1 - |\nu|^2)^2 \label{eq:def_C_1} \\
    \mathcal{C}_2 &\equiv \frac{1}{\pi^2} \int_{\mathbb{C}^2} \diff^2 q_3 \, \diff^2 q_4 \, e^{ - |q_3|^2 - |q_4|^2} \, \mathcal{B}_{0,1} = \sigma(1 - \sigma) (1 - |\nu|^2) (|q_2|^2 - 1) \label{eq:def_C_2}  \\
    \mathcal{C}_3 &\equiv \frac{1}{\pi^2} \int_{\mathbb{C}^2} \diff^2 q_3 \, \diff^2 q_4 \, e^{ - |q_3|^2 - |q_4|^2} \, \Big(  i \sigma q_1 q_4 - \sigma^2 q_2 q_3 \Big) \mathcal{B}_{1,1} = - \sigma^2 ( |q_1|^2 + |q_2|^2 ) (1 - |\nu|^2) \label{eq:def_C_3} \\
    \mathcal{C}_4 &\equiv \frac{1}{\pi^2} \int_{\mathbb{C}^2} \diff^2 q_3 \, \diff^2 q_4 \, e^{ - |q_3|^2 - |q_4|^2} \, \Big(  i \sigma q_1 q_4 - \sigma^2 q_2 q_3 \Big) \mathcal{B}_{0,1} = - \sigma^2 (1 - \sigma) |q_2|^2 (1 - |\nu|^2)\ , \label{eq:def_C_4}
\end{align}
this means we can now say that
\begingroup\allowdisplaybreaks
\begin{align}
    \mathcal{P}&^{(\text{nHIE})}_{N}(z, \bm v; \sigma)= \frac{e^{-(1 - \mu|\nu|^2)|z|^2}}{\pi^{N+3} (1 + \sigma)} \int_{\mathbb{C}^2} \diff^2 q_1 \, \diff^2 q_2  \, e^{ - |q_1|^2 - |q_2|^2} \Bigg[ \Big( |q_1|^2 + \sigma |q_2|^2 + |z|^2 \Big)^{N-2} \mathcal{C}_{1} \nonumber\\
    & + (N-2)  \Big( |q_1|^2 + \sigma |q_2|^2 + |z|^2 \Big)^{N-3} \Big(  |z|^2 \mathcal{C}_{2} + \mathcal{C}_{3}\Big) + (N-2)(N-3) |z|^2 \Big( |q_1|^2 + \sigma |q_2|^2 + |z|^2 \Big)^{N-4}  \mathcal{C}_{4} \Bigg]\ .
    \label{eq:JPDF_q1q2_left}
\end{align}
\endgroup
This necessitates one to be able to perform integrals according to the following identities 
\begin{gather}
    \frac{1}{\pi^2}  \int_{\mathbb{C}^2} \diff^2 q_1 \, \diff^2 q_2  \, e^{ - |q_1|^2 - |q_2|^2} \Big( |q_1|^2 + \sigma |q_2|^2 + |z|^2 \Big)^{m} = e^{|z|^2} m! \, \mathcal{S}_{m}(|z|^2;\sigma) \label{eq:G_0_function} \\
    \frac{1}{\pi^2}  \int_{\mathbb{C}^2} \diff^2 q_1 \, \diff^2 q_2  \, e^{ - |q_1|^2 - |q_2|^2} \Big( |q_1|^2 + \sigma |q_2|^2 + |z|^2 \Big)^{m} |q_1|^2 = e^{|z|^2} m! \Big[ \mathcal{T}_{m+2}(|z|^2;\sigma) - (|z|^2+1) \mathcal{S}_{m+1}(|z|^2;\sigma) \Big] \label{eq:G_1_function} \\
    \frac{1}{\pi^2}  \int_{\mathbb{C}^2} \diff^2 q_1 \, \diff^2 q_2  \, e^{ - |q_1|^2 - |q_2|^2} \Big( |q_1|^2 + \sigma |q_2|^2 + |z|^2 \Big)^{m} |q_2|^2 = e^{|z|^2} m! \Big[ (m+1) \mathcal{S}_{m}(|z|^2;\sigma) - \mathcal{T}_{m}(|z|^2;\sigma) \Big] \ , \label{eq:G_2_function}
\end{gather}
where $m$ is a positive integer and $\mathcal{S}_m(|z|^2; \sigma)$ and $\mathcal{T}_m(|z|^2; \sigma)$ are defined in Eq. \eqref{eq:def_SN_TN}. For convenience of the reader, we provide guidance on how to compute these integrals in Appendix \ref{app:incmpl_gamma_type_ident}. Utilsing Eqs. \eqref{eq:G_0_function} - \eqref{eq:G_2_function} in Eq. \eqref{eq:JPDF_q1q2_left} yields a valid finite-$N$ form of our desired JPDF. One can then simplify this result through repeated application of the recursion relation 
\begin{equation}
    \sigma\Big( m \, \mathcal{S}_{m-1}(|z|^2; \sigma) - \mathcal{T}_{m-1}(|z|^2; \sigma) \Big) = m \, \mathcal{S}_{m}(|z|^2; \sigma) - \mathcal{T}_{m}(|z|^2; \sigma) \ ,
\end{equation}
which, eventually after some manipulation, yields the desired form of Eq. \eqref{eq:JPDF_z_v_nHIE} in Theorem \ref{thm:JPDF_z_v}. Furthermore, one can then easily arrive at the result of Corollary \ref{cor:rho_finite_N}, by integrating over the eigenvectors through the use of \cite[Eq. (8)]{AFS25}
\begin{equation}
    \frac{1}{\pi^N} \int_{\mathbb{C}^N} \diff^2 \bm v \, \delta(\bm v^\dagger \bm v - 1) \, f(|\bm v^T \bm v|^2) = \int_0^1 \diff p \, \frac{p^{N-2}}{\Gamma(N-1)}f(1 - p^2) \ ,
    \label{eq:integrate_out_vectors}
\end{equation}
and the definition of the lower incomplete $\gamma$-function from Eq. \eqref{eq:def_gamma}.

\section{Asymptotic Analysis}
\label{sec:AA}

In this Section, we aim to prove our asymptotic results for the density of eigenvalues, which we outlined in Corollaries \ref{cor:rho_bulk}, \ref{cor:rho_edge_SA} and \ref{cor:rho_edge_WA}. To this end, we begin by re-expressing our $\mathcal{S}$ and $\mathcal{T}$ functions, which act as building blocks for the density, in forms that are more amenable to asymptotic analysis. Starting with $\mathcal{S}_N$, one can employ the integral representation of the incomplete $\Gamma$-function from Eq. \eqref{eq:incmpl_Gamma} to see that  
\begin{align}
    \mathcal{S}_{N}(|z|^2;\sigma) &= \frac{\sigma^N}{\Gamma(N+1)}   \int_{|z|^2}^\infty \diff x \, e^{-x + \frac{x}{\sigma}} \, \Gamma\left(N+1, \frac{x}{\sigma}\right)
\end{align}
and then we use integration by parts to evaluate this integral and arrive at
\begin{align}
    \mathcal{S}_{N}(|z|^2;\sigma) &= \frac{1}{1-\sigma}\frac{\Gamma\left(N+1,|z|^2\right)}{\Gamma(N+1)} - \frac{\sigma^{N+1}}{1-\sigma} e^{-|z|^2  + \frac{|z|^2}{\sigma}} 
    \frac{\Gamma\left(N+1,\tfrac{|z|^2}{\sigma}\right)}{\Gamma\left(N + 1\right)}\ .
    \label{eq:S_N_integrated}
\end{align}
On the other hand, in order to re-express $\mathcal{T}_N$, we first write it in an integral form, 
\begin{equation}
    \mathcal{T}_N(|z|^2;\sigma) = \frac{\sigma^{N-1}}{\Gamma(N)} \int_{|z|^2}^\infty \diff x \, x \, e^{x \frac{1-\sigma}{\sigma}} \Gamma\left( N, \frac{x}{\sigma} \right) \ , 
\end{equation}
which is achieved via differentiating the summation representation in Eq. \eqref{eq:def_SN_TN}. One can then evaluate this integral, so as to arrive at 
\begin{align}
    \mathcal{T}_N(|z|^2; \sigma) &= \frac{N}{1 - \sigma} \frac{\Gamma(N, |z|^2)}{\Gamma(N)} - \frac{1}{1- \sigma} \left( \frac{\sigma}{1-\sigma} \frac{\Gamma(N, |z|^2)}{\Gamma(N)} - \frac{e^{-|z|^2} |z|^{2N}}{\Gamma(N)} \right) \nonumber\\
    & - \frac{\sigma^N}{(1-\sigma)^2} \frac{\Gamma\left(N,\frac{|z|^2}{\sigma} \right)}{\Gamma(N)} e^{\frac{1- \sigma}{\sigma} |z|^2} \Big( (1- \sigma)|z|^2 - \sigma \Big) \ .
    \label{eq:T_N_integrated}
\end{align}
Considering the form of $\mathcal{S}_N$ and $\mathcal{T}_N$ after removing the summations, one can rewrite our equation for the eigenvalue density, Eq. \eqref{eq:rho_finite_N}, as
\begingroup\allowdisplaybreaks
\begin{align}
    \rho^{(\text{nHIE})}_{N}(z;\sigma) &= \frac{1}{\pi(1-\sigma^2)} \Big( \Theta_{N}(|z|^2) - \mathcal{L}_{N}(|z|^2;\sigma) \Big) - \frac{\sigma}{2 \pi(1 - \sigma)} \Big( \Theta_{N-1}(|z|^2) - \mathcal{L}_{N-1}(|z|^2;\sigma) \Big) \nonumber\\
    &+\frac{e^{\frac{\sigma}{1+\sigma}|z|^2}}{2\pi(1+\sigma)}
    \left(\frac{1+\sigma}{\sigma |z|^2}\right)^{\frac{N+1}{2}}
    \gamma\!\left(\frac{N+1}{2}, \frac{\sigma}{1+\sigma}|z|^2\right) \Bigg[ \frac{(2 - \sigma)\big( N (1+\sigma)^2 - \sigma(2+2|z|^2 +\sigma) - 1 \big)}{2(1 - \sigma^2)} \Theta_{N-1}(|z|^2) \nonumber \\
    & - \frac{(N-1) \sigma (1+\sigma)^2 - 2 |z|^2 (\sigma^3 + \sigma - 1)}{2 (1 - \sigma^2)} \mathcal{L}_{N-1}(|z|^2;\sigma) + \frac{1-\sigma^2}{1-\sigma} \frac{e^{-|z|^2} \, |z|^{2(N-1)}} {\Gamma(N-1)} \nonumber \\
    &- \frac{N (1 + \sigma) - \sigma (1 + |z|^2 (2 + (\sigma - 2) \sigma)) - 1}{1 - \sigma^2} \Theta_{N}(|z|^2) + \frac{N (1 + \sigma) - (1 + \sigma + |z|^2 \sigma^2)}{1 - \sigma^2} \mathcal{L}_{N}(|z|^2;\sigma) \Bigg] \, ,
    \label{eq:rho_Theta_L}
\end{align}
\endgroup
where we have introduced the functions
\begin{equation}
    \Theta_{N}(|z|^2) \equiv \frac{\Gamma(N,|z|^2)}{\Gamma(N)}\hspace{1cm} \text{and} \hspace{1cm} \mathcal{L}_{N}(|z|^2;\sigma) \equiv \sigma^N \frac{\Gamma\left(N, \frac{|z|^2}{\sigma} \right)}{\Gamma(N)}  e^{\frac{1- \sigma}{\sigma} |z|^2} \ .
    \label{eq:Theta_L_defs}
\end{equation}
The advantage of expressing $\rho^{(\text{nHIE})}_{N}(z;\sigma)$ in this manner is that we are able to calculate asymptotic expansions of the contributing parts as $N$ becomes large. To this end, we require some general asymptotic expansions, which we will utilise throughout for general finite $a$, $b$, $c$ and $d$, 
\begin{equation}
    \left( 1+ \frac{a}{\sqrt{N}} + \frac{b}{N} \right)^N e^{- a \sqrt{N} - c - \frac{d}{\sqrt{N}}} = e^{-\frac{a^2}{2} + (b-c)} \left( 1+ \frac{1}{\sqrt{N}} \left( \frac{a^3}{3} - ab - d \right) \right) + O \left( \frac{1}{N} \right) 
    \label{eq:power_expo_2nd_order}
\end{equation}
and
\begin{align}
    \frac{\Gamma(N-m, N+ a \sqrt{N} + b)}{\Gamma(N-m)} = \frac{1}{2} \erfc\left( \frac{a}{\sqrt{2}} \right) + \frac{1}{\sqrt{2 \pi N}} \left( \frac{1}{3} (2+a^2) - (b+m+1) \right) e^{-\frac{a^2}{2}} + O\left( \frac{1}{N} \right) \ ,
    \label{eq:Gamma_ratio_edge_2nd_order}
\end{align}
for fixed $m$ - we refer the reader to Appendix \ref{app:asymp_expansions} for a derivation of these expansions. Furthermore, we note that, from Eq. \eqref{eq:Gamma_ratio_edge_2nd_order}, we can recover the well-known asymptotic behaviour of a ratio of $\Gamma$-functions 
\begin{equation}
    \lim_{N \to \infty} \frac{\Gamma(N-m, N a)}{\Gamma(N-m)} = \Theta[1 - a] \ ,
    \label{eq:bulk_Gamma_limit}
\end{equation}
where $\Theta$ without an indexing subscript indicates the Heaviside $\Theta$-function.

With these preliminaries for asymptotic analysis established, we are now in a position to prove our three separate Corollaries over the course of the following Subsections.

\subsection{Proof of Corollary \ref{cor:rho_bulk}}
\label{subsec:proof_bulk}

Here, we provide a proof of the result of Corollary \ref{cor:rho_bulk}, which shows that, to leading order in the nHIE, the circular law holds in the bulk, for all $\sigma \in [0,1]$ and $z = \sqrt{N} w$. This proof is split into two parts, as we must consider the strong and weak asymmetry regimes of $\sigma$ separately.

\begin{proof}

Starting with the regime of strong asymmetry, where $\sigma \in [0,1)$ is fixed as $N \to \infty$, we can notice that the contributions from $\mathcal{L}_N(|z|^2; \sigma)$ are subleading when compared to $\Theta_N(|z|^2)$. Specifically, for large $N$, one can ascertain that $\mathcal{L}_N(|z|^2; \sigma)$ contributes at order $N^{-1/2}$, whereas from Eq. \eqref{eq:bulk_Gamma_limit} we can see that both $\Theta_{N-1}(|z|^2)$ and $\Theta_N(|z|^2)$ contribute at order unity - this also applies at weak asymmetry. For the sake of brevity, we will skip the derivation of the leading order term in $\mathcal{L}_N(|z|^2; \sigma)$ in this Section, as we show how to analyse its asymptotic expansion when it does contribute to leading order in Section \ref{subsec:proof_edge_WA}. Furthermore, it is also easy to see that the term proportional to $ e^{-|z|^2} \, |z|^{2(N-1)}/ \Gamma(N-1)$ only contributes at order $N^{-1/2}$ - see \cite[Eq. (4.43)]{MJC_Thesis} for an example of this analysis.

With Eqs. \eqref{eq:rho_Theta_L} and \eqref{eq:Theta_L_defs} in place, we can apply the bulk scaling limit to see that 
\begin{align}
    \frac{(2 - \sigma)\big( N (1+\sigma)^2 - \sigma(2+2|z|^2 +\sigma) - 1 \big)}{2(1 - \sigma^2)} \Theta_{N-1}(|z|^2) &- \frac{N (1 + \sigma) - \sigma (1 + |z|^2 (2 + (\sigma - 2) \sigma)) - 1}{1 - \sigma^2} \Theta_{N}(|z|^2) \nonumber\\&= \frac{N\sigma}{2} \big( 1 - 2 \mu |w|^2 \big) \Theta\big[ 1 - |w|^2 \big] + O(1)
\end{align}
and furthermore, employing \cite[8.11.6]{DLMF}, which states that
\begin{equation}
    \gamma(a,z) = -z^{a} e^{-z} \left[ \frac{1}{z-a} +  O \left( \frac{1}{(z-a)^3} \right) \right]\ ,
    \label{eq:DLMF_gamma}
\end{equation}
for large $a$ and $z$, such that they have a fixed ratio $z/a < 1$, we see that the behaviour of the term with the lower-incomplete $\gamma$-function is given by 
\begin{equation}
    \frac{e^{\mu N|w|^2}}{(\mu N|w|^2)^{\frac{N+1}{2}}} \gamma\!\left(\frac{N+1}{2}, \mu |z|^2\right) = - \frac{2}{N(2\mu |w|^2 - 1)} + O \left( \frac{1}{N^3} \right) \ ,
\end{equation}  
since $\mu = \sigma/(1+\sigma) < 1/2$. Therefore, by bringing the previous steps together in Eq. \eqref{eq:rho_Theta_L}, we can see that 
\begin{align}
    \rho^{(\text{nHIE})}_{\text{bulk,SA}}(w) &\equiv \lim_{N \to \infty} \rho^{(\text{nHIE})}_{N} \big(\sqrt{N} w; \sigma \big) \nonumber\\
    &= \frac{1}{\pi(1 - \sigma^2)} \Theta\big[1 - |w|^2\big] - \frac{\sigma}{2 \pi(1 - \sigma)} \Theta\big[1 - |w|^2\big] - \frac{1}{2\pi(1+\sigma)} \frac{2}{(2\mu |w|^2 - 1)}  \frac{\sigma}{2} \left( 1 - 2 \mu |w|^2 \right) \Theta\big[1 - |w|^2\big] \nonumber \\
    &= \frac{1}{\pi} \Theta\big[1 - |w|^2\big] \ ,
    \label{eq:rho_nHIE_bulk_SA}
\end{align}
thus the circular law holds to leading order for all $\sigma \in [0,1)$ at strong asymmetry. Now moving onto the limit of weak asymmetry, where $\sigma = 1 - \kappa N^{-1/2}$, we first of all notice that, for the first terms in the summation of Eq. \eqref{eq:rho_Theta_L},
\begin{equation}
    \frac{1}{1-\sigma^2} - \frac{\sigma}{2(1-\sigma)} =  \frac{3}{4} + O\left(N^{-1/2}\right) \ .
    \label{eq:weak_bulk_pt1}
\end{equation}
Then, considering the term which is proportional to the $\gamma$-function in Eq. \eqref{eq:rho_Theta_L}, we can employ Mathematica to expand the terms which are proportional to $\Theta_{N}$ and $\Theta_{N-1}$ in powers of $N$, this yields
\begin{align}
    \frac{(2 - \sigma)\big( N (1+\sigma)^2 - \sigma(2+2|z|^2 +\sigma) - 1 \big)}{2(1 - \sigma^2)} \Theta_{N-1}(|z|^2)& - \frac{N (1 + \sigma) - \sigma (1 + |z|^2 (2 + (\sigma - 2) \sigma)) - 1}{1 - \sigma^2} \Theta_{N}(|z|^2) \nonumber \\
    &= \frac{N}{2} \big(1 - |w|^2\big) \Theta\big[1 - |w|^2\big] + O(\sqrt{N}) \ .
    \label{eq:weak_bulk_pt2}
\end{align}
Moreover, by applying Eq. \eqref{eq:DLMF_gamma} again, we can see that 
\begin{align}
    \frac{e^{\frac{\sigma}{1+\sigma}|z|^2}}{2\pi(1+\sigma)}
    \left(\frac{1+\sigma}{\sigma |z|^2}\right)^{\frac{N+1}{2}}
    \gamma\!\left(\frac{N+1}{2}, \frac{\sigma}{1+\sigma}|z|^2\right) = \frac{1}{2\pi} \frac{1}{N - N|w|^2} + O\left( \frac{1}{N^3} \right) \ ,
    \label{eq:weak_bulk_pt3}
\end{align}
where we have used that
\begin{equation}
    \frac{\sigma}{1 + \sigma} = \frac{1 - \frac{\kappa}{\sqrt{N}}}{2 - \frac{\kappa}{\sqrt{N}}} = \frac{1}{2} \left( 1 - \frac{\kappa}{2\sqrt{N}}\right) + O\left(\frac{1}{N} \right) \ .
\end{equation}
Finally, by combining the findings from Eqs. \eqref{eq:weak_bulk_pt1} - \eqref{eq:weak_bulk_pt3} we see that
\begin{align}
    \rho^{(\text{nHIE})}_{\text{bulk,WA}}(w) &\equiv \lim_{N \to \infty} \rho^{(\text{nHIE})}_{N} \left(z = \sqrt{N} w; \sigma = 1 - \frac{\kappa}{\sqrt{N}} \right) = \frac{1}{\pi} \, \Theta\big[1 - |w|^2\big] \ ,
    \label{eq:rho_nHIE_bulk_WA}
\end{align}
as expected. Hence, from Eqs. \eqref{eq:rho_nHIE_bulk_SA} and \eqref{eq:rho_nHIE_bulk_WA} we have shown that the circular law holds for all $\sigma\in [0,1]$ in the nHIE - thus completing our proof.

\end{proof}

\subsection{Proof of Corollary \ref{cor:rho_edge_SA}}
\label{subsec:proof_edge_SA}

In what follows, we provide a proof of the edge density of the eigenvalues in the nHIE at strong asymmetry - originally presented in Corollary \ref{cor:rho_edge_SA}.

\begin{proof}
When considering asymptotic behaviour in the edge regime, we scale eigenvalues such that $|z| = \sqrt{N} + \eta$. Within such a scaling, the behaviour of a ratio of $\Gamma$-functions is known via Eq. \eqref{eq:Gamma_ratio_edge_2nd_order}. This makes treatment of the $\Theta$ functions trivial and thus, considering the form of Eq. \eqref{eq:rho_Theta_L}, the main challenge that remains asymptotically is to assess the leading order behaviour of the lower incomplete $\gamma$-function and $\mathcal{L}_N(|z|^2; \sigma)$. To this end, we introduce the object
\begin{equation}
    \Xi_N(\eta; \mu) \equiv  \frac{e^{\mu \big(\sqrt{N} +  \eta\big)^2}}{\big(\mu \big(\sqrt{N} +  \eta\big)^2 \big)^{\frac{N+1}{2}}} \gamma\!\left(\frac{N+1}{2}, \mu \big(\sqrt{N} +  \eta\big)^2\right) \ .
\end{equation}
Considering this function, and the related asymptotic behaviour in Eq. \eqref{eq:DLMF_gamma}, one notices that the transitional effect of the lower order $\eta \sqrt{N}$ term is only seen for $\mu \approx 1/2$ (hence the need for a separate scaling of weak asymmetry, discussed in the next Section), thus we can say that 
\begin{equation}
    \Xi_N\left(\eta; \mu < \frac{1}{2} \right) = - \frac{2}{N(2 \mu - 1)} + O \left(\frac{1}{N^3} \right) = \frac{2(1+\sigma)}{N(1 - \sigma)} + O \left(\frac{1}{N^3} \right)  \ .
\end{equation}
On the other hand, we explicitly consider the form of $\mathcal{L}_N(|z|^2; \sigma)$ in this limit and see that
\begin{equation}
    \mathcal{L}_N\Big( (\sqrt{N} +\eta )^2; \sigma \Big) = \sigma^N e^{\frac{1-\sigma}{\sigma}(\sqrt{N} +\eta )^2} \frac{\Gamma\left( N, \frac{(\sqrt{N} +\eta)^2}{\sigma} \right)}{\Gamma(N)} \ .
\end{equation}
Given that $0\leq \sigma < 1$, we are safe to apply the following identity from \cite[8.11.7]{DLMF}
\begin{equation}
    \Gamma(a,z) = z^{a} e^{-z} \left[ \frac{1}{z-a} +  O \left( \frac{1}{(z-a)^3} \right) \right]\ ,
    \label{eq:DLMF_Gamma}
\end{equation}
which is valid under the same conditions as Eq. \eqref{eq:DLMF_gamma}, and so, upon application of Stirling's formula we arrive at
\begin{equation}
    \mathcal{L}_N\Big( (\sqrt{N} +\eta )^2; \sigma \Big) = \frac{1}{\sqrt{2 \pi N}} \frac{\sigma}{1 - \sigma} \left( 1+ \frac{2 \eta}{\sqrt{N}} + \frac{\eta^2}{N} \right)^N e^{- 2 \eta \sqrt{N} - \eta^2} + O\left( \frac{1}{N} \right)\ .
\end{equation}
One can then employ Eq. \eqref{eq:power_expo_2nd_order}, which yields
\begin{equation}
    \mathcal{L}_N\Big( (\sqrt{N} +\eta )^2; \sigma \Big) = \frac{1}{\sqrt{2 \pi N}} \frac{\sigma}{1 - \sigma} e^{-2\eta^2} + O \left( \frac{1}{N} \right) \ .
\end{equation}
Thus, by comparison to Eq. \eqref{eq:Gamma_ratio_edge_2nd_order}, we see that contributions from terms proportional to $\mathcal{L}_N(|z|^2; \sigma)$ are sub-leading. One can also apply a similar process to see that the contribution from the term proportional to $ e^{-|z|^2} \, |z|^{2(N-1)}/ \Gamma(N-1)$ is also sub-leading. Furthermore, the dominant term in the square brackets proportional to the lower incomplete $\gamma$-function yields 
\begin{align}
    \frac{(2 - \sigma)\big( N (1+\sigma)^2 - \sigma(2+2|z|^2 +\sigma) - 1 \big)}{2(1 - \sigma^2)} \Theta_{N-1}(|z|^2)& - \frac{N (1 + \sigma) - \sigma (1 + |z|^2 (2 + (\sigma - 2) \sigma)) - 1}{1 - \sigma^2} \Theta_{N}(|z|^2) \nonumber \\
    &= \frac{N \sigma}{2} \frac{1 - \sigma}{1 + \sigma} \frac{1}{2} \, \erfc\Big( \sqrt{2} \eta \Big) + O(\sqrt{N}) \ ,
\end{align}
and so, by a very similar logic to Eq. \eqref{eq:rho_nHIE_bulk_SA}, we can say that
\begin{align}
    \rho^{(\text{nHIE})}_{\text{edge,SA}}(\eta) &\equiv \lim_{N \to \infty} \rho^{(\text{nHIE})}_{N}\Big(|z| = \sqrt{N} + \eta ; \sigma \Big) = \frac{1}{2 \pi} \erfc\Big(\sqrt{2} \eta\Big) \ ,
\end{align}
which is exactly the leading order edge behaviour in the Ginibre ensembles. 
\end{proof}

\subsection{Proof of Corollary \ref{cor:rho_edge_WA}}
\label{subsec:proof_edge_WA}

In this Subsection, we prove Corollary \ref{cor:rho_edge_WA}, which pertains to the interpolating edge density of the nHIE at weak asymmetry.

\begin{proof}

In order to accurately obtain the leading order behaviour of the density in the edge regime of the nHIE, one must consider the second order expansions of four different objects. These objects are $\Theta_N(|z|^2)$ and $\mathcal{L}_N(|z|^2; \sigma)$, both defined in Eq. \eqref{eq:Theta_L_defs}, as well as $\Delta \Theta(|z|^2) \equiv \Theta_N(|z|^2) - \Theta_{N-1}(|z|^2)$ and $\Delta \mathcal{L}_N(|z|^2; \sigma) \equiv \mathcal{L}_N(|z|^2; \sigma) - \mathcal{L}_{N-1}(|z|^2; \sigma)$. Firstly, we can directly utilise Eq. \eqref{eq:Gamma_ratio_edge_2nd_order} to see that 
\begin{equation}
    \Theta_N\Big( (\sqrt{N} + \eta)^2 \Big) = \frac{\Gamma(N, N + 2 \eta \sqrt{N} + \eta^2)}{\Gamma(N)} = \frac{1}{2} \, \erfc\Big( \sqrt{2} \eta \Big) + \frac{1}{\sqrt{2 \pi N}} e^{-2 \eta^2} \frac{\eta^2 - 1}{3} + O\left( \frac{1}{N} \right) \ .
    \label{eq:Theta_leading_WA}
\end{equation}
Secondly, we can see that
\begin{align}
    \mathcal{L}&_N^{(\text{edge})}(\eta; \kappa) \equiv \mathcal{L}_N\left( N + 2 \eta \sqrt{N}; 1 - \frac{\kappa}{\sqrt{N}}\right) \nonumber \\ 
    &= \left( 1- \frac{\kappa}{\sqrt{N}} \right)^N \exp\left[ \kappa \sqrt{N} + \kappa(2 \eta + \kappa) + \frac{\kappa(\eta+ \kappa^2)}{\sqrt{N}} \right] \frac{\Gamma\left( N,  N + (2 \eta + \kappa) \sqrt{N} + (\eta + \kappa)^2 \right)}{\Gamma(N)} + O\left( \frac{1}{N} \right) \ ,  
\end{align}
where we have utilised that
\begin{equation}
    \frac{N + 2 \eta \sqrt{N} + \eta^2}{1 -  \frac{\kappa}{\sqrt{N}}} = N+ (2 \eta + \kappa) \sqrt{N} + (\eta + \kappa)^2 + O\left(N^{-1/2}\right) \ .
\end{equation}
One can then employ Eqs. \eqref{eq:power_expo_2nd_order} and \eqref{eq:Gamma_ratio_edge_2nd_order} to see that
\begin{align}
    \mathcal{L}_N^{(\text{edge})}(\eta; \kappa) &= e^{2 \eta \kappa + \frac{\kappa^2}{2}} \left( 1+ \frac{\kappa(\eta + \kappa)^2 - \frac{\kappa^3}{3}}{\sqrt{N}}  \right) \Bigg( \frac{1}{2}  \erfc\left( \frac{ 2 \eta + \kappa }{\sqrt{2}} \right)  + \frac{1}{\sqrt{2\pi N}} \frac{\eta^2 -2 \eta  \kappa -2 \kappa^2 -1}{3}  e^{-\frac{(2 \eta + \kappa)^2}{2}} \Bigg) + O\left( \frac{1}{N} \right)\nonumber\\ 
    &= \frac{1}{2} \, e^{2\eta\kappa + \frac{\kappa^2}{2}} \erfc\left( \frac{2 \eta + \kappa}{\sqrt{2}} \right) + \frac{1}{\sqrt{N}} \bigg[ \frac{1}{2} \left(\kappa  (\eta +\kappa )^2-\frac{\kappa ^3}{3}\right) e^{2 \eta  \kappa +\frac{\kappa ^2}{2}} \erfc\left(\frac{2 \eta +\kappa }{\sqrt{2}}\right) \nonumber\\
    & + \frac{1}{\sqrt{2 \pi }} e^{-2 \eta^2} \frac{ \left(\eta ^2-2 \eta  \kappa -2 \kappa ^2-1\right)}{3} \bigg] + O\left( \frac{1}{N} \right)\ .
     \label{eq:L_leading_WA}
\end{align}
Next, we consider $\Delta \Theta_N$, and to this end we introduce
\begin{equation}
    \Delta \Theta^{(\text{WA})} \equiv \Theta_{N}\left( (\sqrt{N} + \eta )^2 \right) - \Theta_{N-1}\left( (\sqrt{N} + \eta )^2 \right) = \frac{(N + 2 \eta \sqrt{N} + \eta^2)^{N-1} e^{-N - 2 \eta \sqrt{N} - \eta^2}}{\Gamma(N)} 
    \label{eq:Delta_Theta_WA_finite} \ ,
\end{equation}
where in the second equality we have used the definition of $\Theta_N(|z|^2)$ in Eq. \eqref{eq:Theta_L_defs} and the recurrence relation \cite[8.356.2]{Gradshteyn},
\begin{equation}
   m \, \Gamma\left( m , x \right) = \Gamma\left( m+1, x \right) - e^{-x} x^{m} \ .
   \label{eq:Gamma_Recursion}
\end{equation}
Therefore, after applying Stirling's formula and Eq. \eqref{eq:power_expo_2nd_order}, we arrive at
\begin{equation}
    \Delta \Theta^{(\text{WA})} = \frac{1}{\sqrt{2 \pi N}} e^{- 2 \eta^2} \left[ 1 + \frac{1}{\sqrt{N}} \left( \frac{2 \eta^3}{3} - 2 \eta \right) \right] +  O\left(\frac{1}{N^{3/2}} \right) \ .
    \label{eq:Delta_Theta_WA}
\end{equation}
Finally, we consider the expansion in $N$ of $\Delta\mathcal{L}$, which we initially write as
\begin{align}
    \Delta \mathcal{L}(|z|^2 ; \sigma) &\equiv \mathcal{L}_N(|z|^2 ; \sigma) - \mathcal{L}_{N-1}(|z|^2 ; \sigma) = \sigma^{N-1} e^{ \frac{1 - \sigma}{\sigma} |z|^2 } \left[ \sigma \frac{\Gamma\left( N, \frac{|z|^2}{\sigma} \right)}{\Gamma(N)} - \frac{\Gamma\left( N-1, \frac{|z|^2}{\sigma} \right)}{\Gamma(N-1)} \right] \ ,
\end{align}
where we have again utilised Eq. \eqref{eq:Gamma_Recursion}. We then apply the edge scaling in the limit of weak asymmetry to see that  
\begin{align}
    \sigma^{N-1} e^{ \frac{1 - \sigma}{\sigma} |z|^2 } &= e^{2 \eta \kappa + \frac{\kappa^2}{2}} \left( 1+ \frac{\kappa(\eta + \kappa)^2 - \frac{\kappa^3}{3} + \kappa}{\sqrt{N}} \right) + O\left( \frac{1}{N} \right)  \ ,
    \label{eq:Delta_L_1}
\end{align}
which is obtained through the use of Eq. \eqref{eq:power_expo_2nd_order}. Furthermore, one can apply the recursion relation once more to see that
\begin{equation}
    \sigma \frac{\Gamma\left( N, \frac{|z|^2}{\sigma} \right)}{\Gamma(N)} - \frac{\Gamma\left( N-1, \frac{|z|^2}{\sigma} \right)}{\Gamma(N-1)} = \frac{e^{-\frac{|z|^2}{\sigma}} \left( \frac{|z|^2}{\sigma} \right)^{N-1} }{\Gamma(N)} - \frac{\kappa}{\sqrt{N}} \frac{\Gamma\left( N, \frac{|z|^2}{\sigma} \right)}{\Gamma(N)}
\end{equation}
where the two contributing terms can be asymptotically expanded as
\begin{align}
    \frac{e^{-\frac{|z|^2}{\sigma}} \left( \frac{|z|^2}{\sigma} \right)^{N-1} }{\Gamma(N)} &= \frac{1}{\sqrt{2\pi N}}  e^{-\frac{(2 \eta + \kappa)^2}{2}} \left( 1 + \frac{(2 \eta + \kappa)(\eta^2 - 2 \eta \kappa - 2\kappa^2 -3 )}{3 \sqrt{N}} \right) + O \left( \frac{1}{N^{3/2}} \right) 
    \label{eq:Delta_L_2}
\end{align} 
and
\begin{align}
    \frac{\Gamma\left( N, \frac{|z|^2}{\sigma} \right)}{\Gamma(N)} &= \frac{1}{2} \erfc\left( \frac{2 \eta + \kappa}{\sqrt{2}} \right) + \frac{1}{\sqrt{2 \pi N}} \frac{\eta^2 - 2 \eta \kappa -2\kappa^2 -1}{3} e^{-\frac{(2 \eta + \kappa)^2}{2}} + O \left( \frac{1}{N} \right) \ ,
    \label{eq:Delta_L_3}
\end{align}
through the use of Stirling's formula and Eqs. \eqref{eq:power_expo_2nd_order} and \eqref{eq:Gamma_ratio_edge_2nd_order}. One can then combine and manipulate the findings of Eqs. \eqref{eq:Delta_L_1} - \eqref{eq:Delta_L_3} to see that
\begin{align}
     \Delta \mathcal{L}^{(\text{WA})} &\equiv \Delta \mathcal{L}\left((\sqrt{N} + \eta)^2 ; 1- \frac{\kappa}{\sqrt{N}} \right)\nonumber\\ 
     &= \frac{1}{\sqrt{N}} \left[ \frac{1}{\sqrt{2 \pi}} e^{-2 \eta ^2} - \frac{\kappa}{2} e^{2 \eta \kappa + \frac{\kappa^2}{2}} \erfc\left( \frac{2 \eta + \kappa}{\sqrt{2}} \right) \right] + \frac{1}{N} \bigg[ - \frac{\kappa}{2} \left( \kappa(\eta + \kappa)^2 - \frac{\kappa^3}{3} + \kappa \right) e^{2 \eta \kappa + \frac{\kappa^2}{2}} \erfc\left( \frac{2 \eta + \kappa}{\sqrt{2}}\right) \nonumber \\
     & + \frac{1}{\sqrt{2 \pi}} e^{-2 \eta^2} \frac{2 \eta ^3-\eta ^2 \kappa +2 \eta  \left(\kappa ^2-3\right)+2 \kappa ^3+\kappa}{3}  \bigg] + O \left( \frac{1}{N^{3/2}} \right) \ .
     \label{eq:Delta_L_WA}
\end{align}
Now that we have accurate asymptotic expansions in place for our building block quantities, we are in a position to start to write down leading order contributions to the eigenvalue density in the edge regime at weak asymmetry. We do this in three separate stages, starting with
\begin{equation}
    T_1(|z|^2;\sigma) \equiv \frac{1}{\pi(1-\sigma^2)} \Big( \Theta_{N}(|z|^2) - \mathcal{L}_{N}(|z|^2;\sigma) \Big) - \frac{\sigma}{2 \pi(1 - \sigma)} \Big( \Theta_{N-1}(|z|^2) - \mathcal{L}_{N-1}(|z|^2;\sigma) \Big) \ ,
\end{equation}
which in the scaling limit of interest becomes
\begin{align}
    T_1\left( (\sqrt{N} + \eta)^2 ; 1 - \frac{\kappa}{\sqrt{N}} \right) &= \frac{1}{\pi} \left[ \frac{\sqrt{N}}{2\kappa} \Delta \Theta^{(1)} - \frac{\sqrt{N}}{2\kappa} \Delta \mathcal{L}^{(1)} + \frac{3}{4} \Theta_{N}^{(1)} - \frac{3}{4} \mathcal{L}_{N}^{(1)} \right] + O\left( \frac{1}{\sqrt{N}} \right) \nonumber\\
    &= \frac{1}{\pi}\left[ \frac{3}{8} \erfc(\sqrt{2} \eta) - \frac{1}{8} e^{2 \eta \kappa + \frac{\kappa^2}{2}} \erfc \left( \frac{2 \eta + \kappa}{\sqrt{2}} \right) \right] + O\left( \frac{1}{\sqrt{N}} \right) \ ,
    \label{eq:T_1_leading} 
\end{align}
where the superscript ${(1)}$ represents the first order term in $N$ etc. from the associated expansion. Next, we consider 
\begin{equation}
    T_2(|z|^2 ; \sigma) \equiv \frac{e^{\frac{\sigma}{1+\sigma}|z|^2}}{2\pi(1+\sigma)}
    \left(\frac{1+\sigma}{\sigma |z|^2}\right)^{\frac{N+1}{2}}
    \gamma\!\left(\frac{N+1}{2}, \frac{\sigma}{1+\sigma}|z|^2\right) \ ,
\end{equation}
then utilising \cite[Eq. (2.5)]{NOD19} and the fact that $\sigma|z|^2/(1+\sigma) = (N/2) + (\eta - \kappa/4) \sqrt{N} + \cdots$, we arrive at
\begin{align}
    T_2 \left(  (\sqrt{N} + \eta)^2 ; 1 - \frac{\kappa}{\sqrt{N}} \right) = \frac{1}{4 \pi} \sqrt{\frac{\pi}{N}} \, e^{\left( \eta - \frac{\kappa}{4} \right)^2} \, \erfc\bigg( - \left( \eta - \frac{\kappa}{4} \right) \bigg) + O \left( \frac{1}{N} \right)  \ .
    \label{eq:T_2_weak_edge_leading}
\end{align}
The third and final term that we need to consider arises from the large square bracket in Eq. \eqref{eq:rho_Theta_L}, i.e.
\begin{align}
    T_3(|z|^2 ; \sigma) &\equiv - \frac{N (1 + \sigma) - \sigma (1 + |z|^2 (2 + (\sigma - 2) \sigma)) - 1}{1 - \sigma^2} \Theta_{N}(|z|^2) + \frac{N (1 + \sigma) - (1 + \sigma + |z|^2 \sigma^2)}{1 - \sigma^2} \mathcal{L}_{N}(|z|^2;\sigma)\nonumber\\
    & + \frac{(2 - \sigma)\big( N (1+\sigma)^2 - \sigma(2+2|z|^2 +\sigma) - 1 \big)}{2(1 - \sigma^2)} \Theta_{N-1}(|z|^2) \nonumber \\
    & - \frac{(N-1) \sigma (1+\sigma)^2 - 2 |z|^2 (\sigma^3 + \sigma - 1)}{2 (1 - \sigma^2)} \mathcal{L}_{N-1}(|z|^2;\sigma) + \frac{1-\sigma^2}{1-\sigma} \frac{e^{-|z|^2} \, |z|^{2(N-1)}} {\Gamma(N-1)} \ .
\end{align}
To facilitate this analysis, we first notice that the final term can be written as
\begin{equation}
     \frac{1-\sigma^2}{1-\sigma} \frac{e^{-|z|^2} \, |z|^{2(N-1)}} {\Gamma(N-1)} = \sqrt{\frac{2N}{\pi}} e^{- 2\eta^2} + O(1) \ ,
\end{equation}
where we have used Eq. \eqref{eq:power_expo_2nd_order} in conjunction with Stirling's formula. To analyse the other 4 terms, we start by expanding all the coefficients in powers of $N$, such that
\begin{align}
    - \frac{N (1 + \sigma) - \sigma (1 + |z|^2 (2 + (\sigma - 2) \sigma)) - 1}{1 - \sigma^2} \Theta_{N} = \bigg[  -\frac{N^{3/2}}{2 \kappa} + \frac{1}{\kappa} \left( \eta - \frac{\kappa}{4} \right) N& +\frac{ \left(4 \eta ^2-4 \eta  \kappa +3 \kappa ^2+8\right)}{8 \kappa } \sqrt{N} \bigg] \Theta_{N} +O(1)
    \label{eq:coeff_Theta_N_expanded}
\end{align}
\begin{align}
    + \frac{(2 - \sigma)\big( N (1+\sigma)^2 - \sigma(2+2|z|^2 +\sigma) - 1 \big)}{2(1 - \sigma^2)} \Theta_{N-1}
    =\bigg[ \frac{N^{3/2}}{2 \kappa} - \frac{1}{\kappa} \left( \eta - \frac{\kappa}{4} \right) N& - \frac{(2 \eta + \kappa)^2 + 8}{8 \kappa} \sqrt{N}\bigg] \Theta_{N-1} + O(1)
    \label{eq:coeff_Theta_Nm1_expanded}
\end{align}
\begin{align}
    - \frac{(N-1) \sigma (1+\sigma)^2 - 2 |z|^2 (\sigma^3 + \sigma - 1)}{2 (1 - \sigma^2)} \mathcal{L}_{N-1} = \bigg[ -\frac{N^{3/2}}{2 \kappa } +\frac{1}{\kappa}\left(\eta -\frac{\kappa}{4}\right) N& +\frac{\left(4 \eta ^2-28 \eta  \kappa +\kappa ^2+8\right)}{8 \kappa } \sqrt{N} \bigg] \mathcal{L}_{N-1} + O(1)
    \label{eq:coeff_L_Nm1_expanded}
\end{align}
\begin{align}
    + \frac{N (1 + \sigma) - (1 + \sigma + |z|^2 \sigma^2)}{1 - \sigma^2} \mathcal{L}_{N} = \bigg[ \frac{N^{3/2}}{2 \kappa } - \frac{1}{\kappa} \left(\eta - \frac{3}{4} \kappa \right) N& + \frac{\left(-4 \eta ^2+12 \eta  \kappa -\kappa ^2-8\right)}{8 \kappa } \sqrt{N} \bigg] \mathcal{L}_{N} +O(1) \ ,
    \label{eq:coeff_L_N_expanded}
\end{align}
where we have suppressed functional arguments. Given that Eq. \eqref{eq:T_2_weak_edge_leading} scales as $N^{-1/2}$, the sum of Eqs. \eqref{eq:coeff_Theta_N_expanded} - \eqref{eq:coeff_L_N_expanded} must scale as $\sqrt{N}$ so as to ensure that the contribution to the density is of order unity. It is easy to see that the contribution at the level of $N^{3/2}$ is zero, but we must also verify that the contribution on the scale of $N$ vanishes. To this end we appeal to Eqs. \eqref{eq:L_leading_WA}, \eqref{eq:Delta_Theta_WA} and \eqref{eq:Delta_L_WA} to see that
\begin{align}
    &\frac{N^{3/2}}{2 \kappa} \Delta \mathcal{L}^{(1)} - \frac{N^{3/2}}{2 \kappa} \Delta\Theta^{(1)} + \frac{N}{2} \mathcal{L}_N^{(1)} = 0 \ ,
\end{align}
and so we now turn our attention to calculating the term which is proportional to $\sqrt{N}$. If one carefully collects powers of $N$ when summing Eqs. \eqref{eq:coeff_Theta_N_expanded} - \eqref{eq:coeff_L_N_expanded}, we find that the term which is proportional to $\sqrt{N}$ is given by
\begin{align}
    T_3 &= - \frac{N^{3/2}}{2 \kappa} \Delta \Theta^{(2)} + \frac{N^{3/2}}{2 \kappa} \Delta \mathcal{L}^{(2)} + N \widetilde{\eta} \Delta \Theta^{(1)} - N \widetilde{\eta} \Delta \mathcal{L}^{(1)} + \frac{N}{2} \mathcal{L}_N^{(2)} - \left( \eta - \frac{\kappa}{4} \right) \sqrt{N} \Theta_N^{(1)} - 2 \eta \sqrt{N} \mathcal{L}_N^{(1)} \ ,
\end{align}
where we have introduced $\widetilde{\eta} \equiv (\eta - \kappa/4)/\kappa$. Then, identifying the corresponding terms in Eqs. \eqref{eq:Theta_leading_WA}, \eqref{eq:L_leading_WA}, \eqref{eq:Delta_Theta_WA} and \eqref{eq:Delta_L_WA}, we see that
\begin{align}
    T_3 = \sqrt{N} \left[ \sqrt{\frac{2}{\pi}} \, e^{-2 \eta^2} - \frac{1}{2} \left( \eta - \frac{\kappa}{4} \right) \erfc\left( \sqrt{2} \eta \right)- \frac{4 \eta + 3\kappa}{8} \, e^{2\eta\kappa + \frac{\kappa^2}{2}} \erfc\left( \frac{2 \eta + \kappa}{\sqrt{2}} \right) \right] + O(1)\ ,
    \label{eq:T_3_leading}
\end{align}
therefore, by combining the contributions from $T_1$, $T_2$ and $T_3$, given in Eqs. \eqref{eq:T_1_leading}, \eqref{eq:T_2_weak_edge_leading} and \eqref{eq:T_3_leading}, we arrive at Eq. \eqref{eq:rho_WA_edge} and complete the proof.

\end{proof}

Considering the form of Eq. \eqref{eq:rho_WA_edge}, we can see that for $\kappa = 0$ it readily reproduces the known form in the limit of complex symmetric matrices, Eq. \eqref{eq:sym_edge_density}. Furthermore, if we take the limit of large $\kappa$, we can see that 
\begin{equation}
    e^{\left( \eta - \frac{\kappa}{4} \right)^2} \, \bigg( 1+ \erf\left( \eta - \frac{\kappa}{4} \right) \bigg) = e^{\left( \eta - \frac{\kappa}{4} \right)^2} \, \erfc\left( \frac{\kappa}{4} - \eta \right) = \frac{1}{\sqrt{\pi} \left( \frac{\kappa}{4} - \eta \right)} + O\left( \frac{1}{\kappa^3} \right)
\end{equation}
and
\begin{equation}
    e^{2\eta\kappa + \frac{\kappa^2}{2}} \erfc\left( \frac{2 \eta + \kappa}{\sqrt{2}} \right) = \sqrt{\frac{2}{\pi}} \, e^{-2 \eta^2} \, \frac{1}{2 \eta + \kappa}+ O\left( \frac{1}{\kappa^3} \right)
\end{equation}
where we have used the large argument expansion of the complementary error function, $\erfc(x) \approx e^{-x^2}/(\sqrt{\pi} x)$ for $x \gg 1$ \cite[7.12 (i)]{DLMF}. Applying the above two equations to Eq. \eqref{eq:rho_WA_edge}, we can then see that
\begin{equation}
    \lim_{\kappa\to \infty}  \rho^{(\text{nHIE})}_{\text{edge,WA}}(\eta;\kappa) = \frac{1}{2 \pi} \erfc \left( \sqrt{2} \eta \right)
\end{equation}
and thus, as expected, for large $\kappa$ we arrive back at the Ginibre limit.

\begin{acknowledgments}
M.J.C. is grateful for support from the Heilbronn Institute for Mathematical Research while performing this research and writing the manuscript. The numerical data used in FIG. \ref{fig:edge_densities} was produced using the computational facilities of the \href{http://www.bristol.ac.uk/acrc/}{Advanced Computing Research Centre}, University of Bristol.
\end{acknowledgments}

\appendix

\section{Verification of Ensemble Average Identity}
\label{app:Ensemble_average_identity}

In this Appendix, we provide a short proof of the result of Proposition \ref{prop:ensemble_ave}.
\begin{proof}
Starting from
\begin{align}
    \mathcal{A}(A,B; \tau) &\equiv \left \langle \exp[\Tr(X_N A + X_N^{\dagger}B)] \right \rangle_{X_N} = \left \langle \exp[\Tr( a_\tau G_N A + b_\tau G_N^T A + a_\tau G_N^\dagger B + b_\tau \overline{G_N} B )] \right \rangle_{G_N} \ , 
\end{align}
with $a_\tau \equiv \sqrt{(1+\tau)/2}$ and $b_\tau \equiv \sqrt{(1-\tau)/2}$. Explicitly evaluating the traces we see that
\begin{equation}
    \mathcal{A}(A,B; \tau) = \prod_{i,j=1}^N \Big \langle \exp[ G_{ij} \left( a_\tau A_{ji} + b_\tau A_{ij}\right) + \overline{G_{ij}} \left( a_\tau B_{ij} + b_\tau B_{ji} \right) ] \Big \rangle_{G_{ij}} \ ,
\end{equation}
then using the Gaussian identity of $\int_{\mathbb{C}} \diff^2 z \, \exp[-|z|^2 + z \alpha + \overline{z}\beta] = \pi\exp[\alpha \beta]$, for general $\alpha$ and $\beta$, we see that
\begin{equation}
    \mathcal{A}(A,B; \tau)= \exp[ \sum_{i,j=1}^N \left( a_\tau A_{ji} + b_\tau A_{ij}\right) \left( a_\tau B_{ij} + b_\tau B_{ji} \right) ] \ .
\end{equation}
One may then utilise that $a_\tau b_\tau = \sqrt{1 - \tau^2}/2$ and evaluate the summations to arrive at the result Eq. \eqref{eq:ave_expo_trace_nHIE}.
\end{proof}

\section{Reducing the Pfaffian problem}
\label{app:simplifying_our_Pfaffian}

In order to derive the JPDF of eigenvalue and eigenvector in the nHIE, one must be able to evaluate a Pfaffian of the following form
\begin{equation}
    \mathcal{F} \equiv \Pf\left( \begin{matrix}
        0 & \mathcal{S}_\psi & \mathcal{K}_{11} & i\mathcal{K}_{12} \\
        - \mathcal{S}_\psi^T & 0 & i\mathcal{K}_{21} & \mathcal{K}_{22} \\
        - \mathcal{K}_{11}^T & - i\mathcal{K}_{21}^T & 0 & \mathcal{S}_\varphi \\
        - i\mathcal{K}_{12}^T & - \mathcal{K}_{22}^T & - \mathcal{S}_\varphi^T & 0
    \end{matrix}  \right)
    \label{eq:our_pfaffian_app}
\end{equation}
where the specific form of each of the $N \times N$ block's entries are specified in Eqs. \eqref{eq:K_11} - \eqref{eq:S_psi_def}. For our current purposes, it will suffice to just use abstract forms of their entries, i.e.
\begin{align}
    \mathcal{S}_\psi = \alpha_\psi \1_N + \beta_\psi \overline{\bv} \bv^T \, , \hspace{0.5cm} 
    &\mathcal{S}_\varphi = \alpha_\varphi \1_N + \beta_\varphi \bv \bv^\dagger \, ,  \hspace{0.5cm}
    \mathcal{K}_{11} = \alpha_{11} \1_N + \RV \Lambda_{11} \RV^T \, ,  \hspace{0.5cm}  
    \mathcal{K}_{22} = \alpha_{22} \1_N + \RV \Lambda_{22} \RV^T \nonumber\\
    &\mathcal{K}_{12} = w \1_N + \RV \Lambda_{12} \RV^T \, ,  \hspace{0.5cm} \mathcal{K}_{21} = \overline{w} \1_N + \RV \Lambda_{21} \RV^T  \ ,
\end{align}
with $\RV = (\bv, \overline{\bv})$. To aid analytic progress, we will utilise that $\mathcal{F}^2 = \det(\mathcal{M})$, so as to allow us to use standard identities relating to the determinant. We start by using the following identity for non-commuting blocks
\begin{equation}
    \det\left( \begin{matrix}
        A & B \\ C & D  
    \end{matrix}\right) = \det\Big( D \Big) \det\Big( A - B D^{-1} C \Big) \ ,
\end{equation}
which yields
\begin{equation}
    \mathcal{F}^2 = \det( \begin{matrix}
        0 & \mathcal{S}_\varphi \\
        - \mathcal{S}_\varphi^T & 0
    \end{matrix} ) \det(  \left( \begin{matrix}
        0 & \mathcal{S}_\psi \\ - \mathcal{S}_\psi^T & 0 
    \end{matrix}\right) + \left( \begin{matrix}
        \mathcal{K}_{11} & \mathcal{K}_{12} \\ \mathcal{K}_{21} & \mathcal{K}_{22} 
    \end{matrix} \right)  \left( \begin{matrix}
        0 & \mathcal{S}_\varphi \\
        - \mathcal{S}_\varphi^T & 0
    \end{matrix} \right)^{-1} \left( \begin{matrix}
        \mathcal{K}_{11}^T & \mathcal{K}_{21}^T \\ \mathcal{K}_{12}^T & \mathcal{K}_{22}^T 
    \end{matrix} \right) ) \ .
\end{equation}
One can immediately see that
\begin{equation}
    \det( \begin{matrix}
        0 & \mathcal{S}_\varphi \\
        - \mathcal{S}_\varphi^T & 0
    \end{matrix} ) = \det( \mathcal{S}_\varphi \mathcal{S}_\varphi^T) = \det(\mathcal{S}_\varphi)^2 = \alpha_\varphi^{2(N-1)} (\alpha_\varphi + \beta_\varphi)^2 \ ,
\end{equation}
with the final equality achieved through Sylvester's determinant theorem, Eq. \eqref{eq:Sylvester_det_theorem}. Furthermore, by employing the Sherman-Morrison formula, we know that 
\begin{equation}
    \left(a \1_N + b \bm x \bm y^T\right)^{-1} = \frac{1}{a} \left[ \1_N - \frac{b}{a + b} \bm x \bm y^T \right] \ ,
\end{equation}
for general vectors $\bm x$ and $\bm y$, such that $\bm y^T \bm x = \bm x^T \bm y = 1$, and so we can also see that
\begin{align}
    \left( \begin{matrix}
        0 & \mathcal{S}_\varphi \\
        - \mathcal{S}_\varphi^T & 0
    \end{matrix} \right)^{-1} 
    &= \frac{1}{\alpha_\varphi} \left( \begin{matrix}
        0 & - \1_N + \gamma_\varphi \overline{\bv} \bv^T \\ \1_N - \gamma_\varphi \bv \bv^\dagger & 0 
    \end{matrix} \right) \ ,
\end{align}
where $\gamma_\varphi \equiv \beta_\varphi/(\alpha_\varphi + \beta_\varphi)$. Applying these statements for the determinant inverse, we see that
\begingroup
\allowdisplaybreaks
\begin{align}
    \mathcal{F}^2 &= \left( \frac{\alpha_\varphi + \beta_\varphi}{\alpha_\varphi} \right)^2 \det \left( \left( \begin{matrix}
        0 & \alpha_\varphi \mathcal{S}_\psi \\ - \alpha_\varphi \mathcal{S}_\psi^T & 0 
    \end{matrix}\right) + \left( \begin{matrix}
        \mathcal{K}_{11} & \mathcal{K}_{12} \\ \mathcal{K}_{21} & \mathcal{K}_{22} 
    \end{matrix} \right)
    \left( \begin{matrix}
        0 & - \1_N + \gamma_\varphi \overline{\bv} \bv^T \\ \1_N - \gamma_\varphi \bv \bv^\dagger & 0 
    \end{matrix} \right) \left( \begin{matrix}
        \mathcal{K}_{11}^T & \mathcal{K}_{21}^T \\ \mathcal{K}_{12}^T & \mathcal{K}_{22}^T 
    \end{matrix} \right) \right)
\end{align}
\endgroup
which, after evaluating the products, we can write as
\begin{equation}
    \mathcal{F}^2 = \left( \frac{\alpha_\varphi + \beta_\varphi}{\alpha_\varphi} \right)^2\det \left( \begin{matrix}
        \mathcal{Y}_{11} & \mathcal{Y}_{12} \\ - \mathcal{Y}_{12}^T & \mathcal{Y}_{22}
    \end{matrix}\right) \ ,
    \label{eq:Pfaff_2Nx2N}
\end{equation}
where we have introduced the matrices
\begingroup\allowdisplaybreaks
\begin{align}
    &\mathcal{Y}_{11} = - i \mathcal{K}_{11} \mathcal{K}_{12}^T + i \mathcal{K}_{12} \mathcal{K}_{11}^T + i\gamma_\varphi \mathcal{K}_{11} \overline{\bv} \bv^T \mathcal{K}_{12}^T - i\gamma_\varphi \mathcal{K}_{12} \bv \bv^\dagger \mathcal{K}_{11}^T  \\
    &\mathcal{Y}_{12} = \alpha_\varphi \mathcal{S}_\psi - \mathcal{K}_{11} \mathcal{K}_{22}^T - \mathcal{K}_{12} \mathcal{K}_{21}^T + \gamma_\varphi \mathcal{K}_{11} \overline{\bv} \bv^T \mathcal{K}_{22}^T + \gamma_\varphi \mathcal{K}_{12} \bv \bv^\dagger \mathcal{K}_{21}^T \\
    &\mathcal{Y}_{22} = - i\mathcal{K}_{21} \mathcal{K}_{22}^T + i\mathcal{K}_{22} \mathcal{K}_{21}^T + i\gamma_\varphi \mathcal{K}_{21} \overline{\bv} \bv^T \mathcal{K}_{22}^T - i\gamma_\varphi \mathcal{K}_{22} \bv \bv^\dagger \mathcal{K}_{21}^T  \ .
\end{align}
\endgroup
At this point, we would like to re-express each $\mathcal{Y}$ matrix as any part proportional to the identity plus low rank corrections in $\bv$ and $\overline{\bv}$. We outline the process for re-expressing $\mathcal{Y}_{11}$ and then leave application of similar steps to re-express $\mathcal{Y}_{22}$ and $\mathcal{Y}_{12}$ as an exercise for the reader. Starting with the simpler two terms we see that
\begin{equation}
    i \Big[  \mathcal{K}_{12} \mathcal{K}_{11}^T - \mathcal{K}_{11} \mathcal{K}_{12}^T \Big] = i \RV \bigg( \alpha_{11} \big( \Lambda_{12} - \Lambda_{12}^T \big) + w \big( \Lambda_{11}^T - \Lambda_{11} \big) + \big( \Lambda_{12} \RV^T \RV \Lambda_{11}^T - \Lambda_{11} \RV^T \RV \Lambda_{12}^T \big) \bigg) \RV^T \ . 
    \label{eq:Y11_part1}
\end{equation}
Now considering one of the slightly more involved terms, we expand out
\begin{align}
    i\gamma_\varphi \mathcal{K}_{11} \overline{\bv} \bv^T \mathcal{K}_{12}^T &= i\gamma_\varphi \left( \alpha_{11} \1_N + \RV \Lambda_{11} \RV^T \right) \overline{\bv} \bv^T \left( w \1_N + \RV \Lambda_{12}^T \RV^T \right) \nonumber \\
    &= i\gamma_\varphi \left[ \alpha_{11} w \overline{\bv} \bv^T + w \RV \Lambda_{11} \RV^T \overline{\bv} \bv^T + \alpha_{11} \overline{\bv} \bv^T \RV \Lambda_{12}^T \RV^T + \RV \Lambda_{11} \RV^T  \overline{\bv} \bv^T  \RV \Lambda_{12}^T \RV^T \right] \ .
\end{align}
In an attempt to make analysis simpler at later stages, we now express each of these four contributing terms in the following way
\begin{align}
    \alpha_{11} w \overline{\bv} \bv^T &= \alpha_{11} w \RV \left( \begin{matrix}
        0 & 0 \\ 1 & 0 \\
    \end{matrix} \right) \RV^T \label{eq:rot_rules_start} \\
    w \RV \Lambda_{11} \RV^T \overline{\bv} \bv^T &= w \RV \Lambda_{11} F_1 \RV^T\\
    \alpha_{11} \overline{\bv} \bv^T \RV \Lambda_{12}^T \RV^T &= \alpha_{11} \RV F_2 \Lambda_{12}^T \RV^T \\ 
    \RV \Lambda_{11} \RV^T  \overline{\bv} \bv^T  \RV \Lambda_{12}^T \RV^T &= \RV \Lambda_{11} \mathcal{C} \Lambda_{12}^T \RV^T
    \label{eq:rot_rules_end}
\end{align}
where we have introduced
\begin{align}
    F_1 = \left( \begin{matrix}
        1 & 0 \\ \overline{\nu} & 0
    \end{matrix}\right)
    \hspace{1cm}
    F_2 = \left( \begin{matrix}
        0 & 0 \\ \nu & 1
    \end{matrix}\right)
    \hspace{1cm} \mathcal{C} = \left( \begin{matrix}
        \nu & 1 \\ |\nu|^2 & \overline{\nu}
    \end{matrix}\right) \ .
\end{align}
Therefore,
\begin{align}
    i\gamma_\varphi \mathcal{K}_{11} \overline{\bv} \bv^T \mathcal{K}_{12}^T = i \gamma_\varphi \RV \bigg( \alpha_{11} w \left( \begin{matrix}
        0 & 0 \\ 1 & 0 \\
    \end{matrix} \right) + w \Lambda_{11} F_1 + \alpha_{11} F_2 \Lambda_{12}^T + \Lambda_{11} \mathcal{C} \Lambda_{12}^T \bigg) \RV^T
\end{align}
and so, we can now easily obtain that 
\begin{align}
    + i\gamma_\varphi &\mathcal{K}_{11} \overline{\bv} \bv^T \mathcal{K}_{12}^T - i\gamma_\varphi \mathcal{K}_{12} \bv \bv^\dagger \mathcal{K}_{11}^T  \nonumber \\
    &= i \gamma_\varphi \RV \bigg( \alpha_{11} w \Sigma + w \Big( \Lambda_{11} F_1 -  F_1^T \Lambda_{11}^T \Big) + \alpha_{11} \Big( F_2 \Lambda_{12}^T - \Lambda_{12} F_2^T \Big) + \Lambda_{11} \mathcal{C} \Lambda_{12}^T - \Lambda_{12} \mathcal{C}^T \Lambda_{11}^T \bigg) \RV^T \ ,
    \label{eq:Y11_part2}
\end{align}
where
\begin{equation}
    \Sigma \equiv \left( \begin{matrix}
        0 & -1 \\ 1 & 0 \\
    \end{matrix} \right) = -\Sigma^T\ .
\end{equation}
Therefore, combining Eqs. \eqref{eq:Y11_part1} and \eqref{eq:Y11_part2} we see that 
\begin{equation}
    \mathcal{Y}_{11} = i \RV \mathcal{T}_{11} \RV^T \ ,
\end{equation}
where we have introduced the matrix
\begin{align}
    \mathcal{T}_{11} &\equiv \gamma_\varphi \alpha_{11} w \Sigma + \alpha_{11} \Lambda_{12} \Big( \1_2 - \gamma_\varphi F_2^T \Big) -  \alpha_{11}  \Big( \1_2 - \gamma_\varphi F_2 \Big) \Lambda_{12}^T + w \Lambda_{11} \Big( \gamma_\varphi F_1 - \1_2 \Big) - w \Big( \gamma_\varphi F_1^T - \1_2 \Big) \Lambda_{11}^T \nonumber \\
    & + \Lambda_{12} \Big( \RV^T \RV - \gamma_\varphi \mathcal{C}^T \Big) \Lambda_{11}^T - \Lambda_{11} \Big( \RV^T \RV - \gamma_\varphi \mathcal{C} \Big) \Lambda_{12}^T \ .
    \label{eq:T11_def}
\end{align}
One can then follow a similar procedure to the one outlined above to re-write $\mathcal{Y}_{12}$ and $\mathcal{Y}_{22}$ as
\begin{equation}
    \mathcal{Y}_{22} = i \RV \mathcal{T}_{22} \RV^T \hspace{1cm} \text{and} \hspace{1cm} \mathcal{Y}_{12} = ( \alpha_\varphi \alpha_\psi - \alpha_{11} \alpha_{22} - |w|^2 ) \1_N  +  \RV \mathcal{T}_{12} \RV^T
\end{equation}
such that
\begin{align}
    \mathcal{T}_{22} &\equiv \gamma_\varphi \alpha_{22} \overline{w} \Sigma + \alpha_{22} \Lambda_{21} \Big( \gamma_\varphi F_1 - \1_2 \Big) - \alpha_{22}  \Big( \gamma_\varphi F_1^T - \1_2 \Big) \Lambda_{21}^T + \overline{w} \Big( \gamma_\varphi F_2 - \1_2 \Big)\Lambda_{22}^T - \overline{w} \Lambda_{22} \Big( \gamma_\varphi F_2^T - \1_2 \Big)\nonumber \\
    & + \Lambda_{22} \Big( \RV^T \RV - \gamma_\varphi \mathcal{C}^T \Big) \Lambda_{21}^T - \Lambda_{21} \Big( \RV^T \RV - \gamma_\varphi \mathcal{C} \Big) \Lambda_{22}^T \ .
    \label{eq:T22_def}
\end{align} 
and 
\begin{align}
    \mathcal{T}_{12} &\equiv \left( \begin{matrix}
        0 & \gamma_\varphi |w|^2 \\
        \alpha_\varphi \beta_\psi + \gamma_\varphi \alpha_{11} \alpha_{22} & 0
    \end{matrix}\right) + \alpha_{22} \Lambda_{11} (\gamma_\varphi F_1 - \1_2) + \overline{w} \Lambda_{12} (\gamma_\varphi F_2^T - \1_2) \nonumber \\
    & + w (\gamma_\varphi F_1^T - \1_2) \Lambda_{21}^T + \alpha_{11} (\gamma_\varphi F_2 - \1_2) \Lambda_{22}^T + \Lambda_{11} (\gamma_\varphi \mathcal{C} - \RV^T\RV) \Lambda_{22}^T  + \Lambda_{12} (\gamma_\varphi \mathcal{C}^T - \RV^T\RV) \Lambda_{21}^T \ .
    \label{eq:T12_def}
\end{align} 
Therefore, returning to Eq. \eqref{eq:Pfaff_2Nx2N} and combining the previous findings, we can now re-express our equation for the desired Pfaffian as
\begin{align}
    \mathcal{F}^2 &= \left( \frac{\alpha_\varphi + \beta_\varphi}{\alpha_\varphi} \right)^2 \det[ \left( \begin{matrix}
        0 & \zeta \\ - \zeta & 0 
    \end{matrix} \right) \otimes \1_N + \left( \begin{matrix}
        i \RV \mathcal{T}_{11} \RV^T & \RV \mathcal{T}_{12} \RV^T \\
        - \RV \mathcal{T}_{12}^T \RV^T & i\RV \mathcal{T}_{22} \RV^T
    \end{matrix} \right)] \nonumber \\
    &= \left( \frac{\alpha_\varphi + \beta_\varphi}{\alpha_\varphi} \right)^2 \det[ \left( \begin{matrix}
        0 & \zeta \\ - \zeta & 0 
    \end{matrix} \right) \otimes \1_N + \left( \begin{matrix}
        \RV & 0 \\
       0 & \RV
    \end{matrix} \right) \left( \begin{matrix}
        i \mathcal{T}_{11} & \mathcal{T}_{12} \\
        - \mathcal{T}_{12}^T & i \mathcal{T}_{22}
    \end{matrix} \right) \left( \begin{matrix}
        \RV^T & 0\\
        0 & \RV^T
    \end{matrix} \right)]
\end{align}
with $\zeta \equiv \alpha_\varphi \alpha_\psi - \alpha_{11} \alpha_{22} - |w|^2$. One can now employ the identity
\begin{equation}
    \det[X + AB] = \det[X] \det[\1 + B X^{-1} A ] \ ,
    \label{eq:det_identity}
\end{equation}
for invertible $X$ and appropriately sized $A$ and $B$, which eventually yields that
\begin{align}
    \mathcal{F}^2 &= \left( \frac{\alpha_\varphi + \beta_\varphi}{\alpha_\varphi} \right)^2 \zeta^{2(N-2)} \det[ \zeta \1_4  + \left( \begin{matrix}
        \mathcal{T}_{12} \RV^T \RV & - i \mathcal{T}_{11} \RV^T \RV   \\
        i \mathcal{T}_{22} \RV^T \RV &  \mathcal{T}_{12}^T \RV^T \RV 
    \end{matrix} \right) ] \ ,
    \label{eq:Fsq_4x4}
\end{align}
where the blocks in the second matrix are all proportional to
\begin{equation}
    \RV^T \RV = \left( \begin{matrix}
        \nu & 1 \\ 1 & \overline{\nu}
    \end{matrix} \right) \ .
\end{equation}

\section{Derivation of Incomplete Gamma-Function Identities}
\label{app:incmpl_gamma_type_ident}

In this Appendix, we aim to show how one can arrive at the identities in Eqs. \eqref{eq:G_0_function} - \eqref{eq:G_2_function}, starting with the first one. The key here is to notice that these integrals are only dependent on $|q_1|^2$ and $|q_2|^2$ and so we can employ the change of variables $\diff^2 q = 0.5 \diff R \diff \theta$, where $R = \Re(q)^2 + \Im(q)^2$. Using this on the LHS of Eq. \eqref{eq:G_0_function} we see that
\begingroup\allowdisplaybreaks
\begin{align}
    \frac{1}{\pi^2} \int_{\mathbb{C}^2} \diff^2 q_1 \, \diff^2 q_2 & \, e^{ - |q_1|^2 - |q_2|^2} \Big( |q_1|^2 + \sigma |q_2|^2 + |z|^2 \Big)^{m} = \int_0^\infty \diff R_1 \, e^{-R_1} \int_0^\infty \diff R_2 \, e^{-R_2} \Big( R_1 + \sigma R_2 + |z|^2 \Big)^m \nonumber \\
    &= \sum_{k=0}^m \frac{m!}{(m-k)! k!} \sigma^{m-k} \int_0^\infty \diff R_1 \, e^{-R_1} (R_1 + |z|^2)^k \int_0^\infty \diff R_2 \, e^{-R_2} (\sigma R_2)^{m-k} \ .
\end{align}
\endgroup
One can then evaluate the integral over $R_2$ and use the following representation of the incomplete $\Gamma$-function
\begin{equation}
    \int_0^\infty \diff R \, e^{-R} \, (R + |z|^2)^N = e^{|z|^2} \Gamma(N+1,|z|^2) \ ,
\end{equation}
to see that
\begin{equation}
    \frac{1}{\pi^2} \int_{\mathbb{C}^2} \diff^2 q_1 \, \diff^2 q_2  \, e^{ - |q_1|^2 - |q_2|^2} \Big( |q_1|^2 + \sigma |q_2|^2 + |z|^2 \Big)^{m} = e^{|z|^2} m! \sigma^{m} \sum_{k=0}^m \frac{1}{\sigma^k} \frac{\Gamma(k+1, |z|^2)}{\Gamma(k+1)}
\end{equation}
as required. Furthermore, when trying to prove Eq. \eqref{eq:G_1_function}, we apply similar steps to arrive at
\begingroup\allowdisplaybreaks
\begin{align}
    \frac{1}{\pi^2}  \int_{\mathbb{C}^2} \diff^2 q_1 \, \diff^2 q_2 & \, e^{ - |q_1|^2 - |q_2|^2} \Big( |q_1|^2 + \sigma |q_2|^2 + |z|^2 \Big)^{m} |q_1|^2 = \int_0^\infty \diff R_1 \, R_1 \, e^{-R_1} \int_0^\infty \diff R_2 \, e^{-R_2} \Big( R_1 + \sigma R_2 + |z|^2 \Big)^m \nonumber \\
    &= \sum_{k=0}^m \frac{m!}{(m-k)! k!} \sigma^{m-k} \int_0^\infty \diff R_1 \, R_1 \, e^{-R_1} (R_1 + |z|^2)^k \int_0^\infty \diff R_2 \, e^{-R_2} (\sigma R_2)^{m-k} \ ,
\end{align}
\endgroup
again, the integral over $R_2$ is trivial and then the integral over $R_1$ can be performed using
\begingroup\allowdisplaybreaks
\begin{align}
    \int_0^\infty \diff R \, R \, (R + |z|^2)^N e^{-R} = \frac{e^{|z|^2}}{N+1} \bigg[ \Gamma(N + 2, |z|^2) \big( N + 1 - |z|^2 \big) + (|z|^2)^{N+2} e^{-|z|^2}\bigg] \ .
\end{align}
\endgroup
Simple application of the above formula trivially yields Eq. \eqref{eq:G_1_function}. The process for deriving Eq. \eqref{eq:G_2_function} is essentially identical to deriving Eq. \eqref{eq:G_0_function} and so we omit the details.

\section{Derivation of asymptotic expansions}
\label{app:asymp_expansions}

Here we provide short proofs of this asymptotic expansions provided in Eqs. \eqref{eq:power_expo_2nd_order} and \eqref{eq:Gamma_ratio_edge_2nd_order}, which we utilised widely throughout Section \ref{sec:AA}. Starting with Eq. \eqref{eq:power_expo_2nd_order}, we can carefully take a logarithm to see that 
\begin{equation}
    \ln\left[\left( 1+ \frac{a}{\sqrt{N}} + \frac{b}{N} \right)^N e^{- a \sqrt{N} - c - \frac{d}{\sqrt{N}}}\right] = - \frac{a^2}{2} + (b-c) + \frac{1}{\sqrt{N}} \left( \frac{a^3}{3} - ab - d \right) + O \left( \frac{1}{N} \right) \ ,
\end{equation}
and so Eq. \eqref{eq:power_expo_2nd_order} follows simply. The second expansion requires a little more thought and so we define
\begin{equation}
    \mathcal{R}_{N,m}(a,b) \equiv \frac{\Gamma(N-m, N + a \sqrt{N} + b)}{\Gamma(N-m)}
\end{equation}
and begin by using the integral representation of the incomplete $\Gamma$-function, Eq. \eqref{eq:incmpl_Gamma}, to see that
\begin{equation}
    \mathcal{R}_{N,m}(a,b) = \frac{1}{\Gamma(N-m)} \int_{N + a \sqrt{N} + b}^\infty \diff t \, t^{N-m-1} \, e^{-t} = \frac{1}{\Gamma(N-m)} \int_{a}^\infty \diff z \, \big( N + z \sqrt{N} + b \big)^{N-m-1} \, e^{-(N + z \sqrt{N} + b)} \ ,
\end{equation}
then utilising Stirling's formula, we then find that
\begin{equation}
     \mathcal{R}_{N}(a,b) = \frac{1}{\sqrt{2\pi}} \int_{a}^\infty \diff z \left( 1+ \frac{z}{\sqrt{N}} + \frac{b}{N} \right)^{-m-1} \left( 1+ \frac{z}{\sqrt{N}} + \frac{b}{N} \right)^N e^{- z \sqrt{N} - b} \ .
\end{equation}
We can now directly use Eq. \eqref{eq:power_expo_2nd_order} so as to see that
\begin{align}
    \mathcal{R}_{N,m}(a,b) &= \frac{1}{\sqrt{2\pi}} \int_{a}^\infty \diff z \left( 1 - \frac{(m+1)z}{\sqrt{N}} \right) e^{-\frac{z^2}{2}} \left( 1+ \frac{1}{\sqrt{N}} \left( \frac{z^3}{3} - z b \right) \right) + O \left(  \frac{1}{N} \right) \nonumber\\
    &= \frac{1}{\sqrt{2\pi}} \int_{a}^\infty \diff z \, e^{-\frac{z^2}{2}} \left( 1+ \frac{1}{\sqrt{N}} \left( \frac{z^3}{3} - z (b + m+ 1) \right) \right) + O \left(  \frac{1}{N} \right) \ .
    \label{eq:sqrt_N_finite_expo_2}
\end{align}
One can now collect the terms which are proportional to $\sqrt{N}$ and evaluate the integrals to arrive at our desired result of Eq. \eqref{eq:Gamma_ratio_edge_2nd_order}.

\end{document}